\documentclass[11 pt, reqno]{amsart}
\usepackage[foot]{amsaddr} %para que las filiaciones queden al pie de la primera página

\usepackage{amsmath}
\usepackage{amsthm}
\usepackage[english]{babel}
\usepackage{amssymb}
\usepackage{graphics}
\usepackage{epsfig}
\usepackage{amscd}
\usepackage{amsfonts}
\usepackage{amsbsy}
\usepackage{float}
\usepackage{mathrsfs}
\usepackage{natbib}
\usepackage{booktabs}
\usepackage{caption}
\usepackage{subcaption}
\usepackage{overpic}
\usepackage{dsfont}
\usepackage{natbib}
\usepackage{comment}

\newtheorem {theorem} {Theorem}
\newtheorem {prop} [theorem] {Proposition}

\newtheorem {lemma} [theorem] {Lemma}

\newtheorem {remark} {Remark}

\usepackage{color}%para añadir color jgv & andres

\usepackage{fancyhdr}
\begin{document}

%\title{A Crime/S.I.R. optimal control problem}
%\author{Author 1 $^1$, Author 2 $^1$, Author 3 $^1$ and \\ Author 4 $^2$}
\author{Mariana Álvarez $^1$, Alexander Alegría $^1$, Andr\'es Rivera $^1$ and \\ Sebastián Pedersen $^2$}

%\begin{document}

\title[A Crime/SIR Optimal Control Framework for Youth-Oriented Crime Prevention in Urban Settings]{A Crime/SIR Optimal Control Framework for Youth-Oriented Crime Prevention in Urban Settings}

%\author[John A. Arredondo, Andr\'es Rivera]{John A. Arredondo, Andr\'es Rivera}

%\address{$^1$ Address 1.}
%\address{$^2$ Address 2.}
\address{$^1$ Departamento de Ciencias Naturales y Matem\'aticas Pontificia Universidad Javeriana Cali, Facultad de Ingenier\'ia y Ciencias, Calle 18 No. 118-250 Cali, Colombia.}
\address{$^2$ Departamento de Matem\'atica, Facultad de Ciencias Exactas y Naturales, Universidad de Buenos Aires. Ciudad Universitaria Pab. I, (1428) Buenos Aires, Argentina.}

%\email{email 1, email 2}\emph{}
%\email{email 3, email 4}\emph{}
\email{maalvarezcano1123@javerianacali.edu.co,aalegria@javerianacali.edu.co}\emph{}
\email{amrivera@javerianacali.edu.co,spedersen@dm.uba.ar}\emph{}



\date{}
\dedicatory{}
\maketitle

\begin{center}\rule{0.9\textwidth}{0.1mm}
\end{center}
\begin{abstract}
This paper proposes a mathematical framework based on optimal control theory to analyze youth-related criminal dynamics in urban settings. Crime is modeled through an SIR-type compartmental system describing transitions between susceptible, criminally involved, and rehabilitated populations. Public interventions are formalized via the optimal control problem
\begin{equation*}
\min_{u_1,u_2,u_3}\int_{0}^{t_{\text{F}}} \left( I(t) -
R(t) + \frac{\mathscr{C}_1}{2} u_1^2(t) + \frac{\mathscr{C}_2}{2} u_2^2(t) + \frac{\mathscr{C}_3}{2} u_3^2(t) \right) \, dt,
\end{equation*}
subject to the SIR dynamics
\begin{equation*}
\left\{
\begin{aligned}
    \dot{S} &= \Lambda - (1-u_1)h(S)I - \mu S + ((1+u_3)\gamma_2)I + \rho \Omega R, \\
    \dot{I} &= (1-u_1)h(S)I - (\mu + \delta_1)I - ((1+u_2)\gamma_1)I - ((1+u_3)\gamma_2)I + (1-\Omega)\rho R,\\
    \dot{R} &= ((1+u_2)\gamma_1)I - (\mu + \delta_2 + \rho)R.
\end{aligned}
\right.
\end{equation*}
where the control variables represent preventive, punitive, and social reintegration policies under resource constraints. Analytical results establish positivity, invariance, equilibrium exis\-tence, and local stability properties, including the characterization of crime-free and crime-endemic equilibria through a threshold parameter. As a case study, the model and optimal policies are numerically validated using data from En la Buena!, a youth-oriented government program currently implemented by the Municipality of Santiago de Cali, Colombia, illustrating the cost-effectiveness of integrated intervention strategies and how optimal policy combinations can reduce youth involvement in criminal dynamics in urban contexts marked by structural inequality.

%The proposed model is applied to the youth-oriented government program En la Buena! implemented by the Municipality of Santiago de Cali-Colombia, and numerical simulations calibrated using program data are employed to evaluate the dynamic impact and cost-effectiveness of alternative intervention strategies. The results illustrate how optimal policy combinations can reduce youth involvement in criminal dynamics in urban contexts marked by structural inequality.
\end{abstract}

\section{Introduction}

Violent crime remains one of the most persistent social challenges in Colombia, with particularly severe consequences in large urban centers. According to official records from \citet{tasaHomicidios2025} and the National Police, more than 13,000 homicides were reported throughout the country in 2024, of which approximately $950$ occurred in \textit{Santiago de Cali} (from now on Cali), representing more than $7\%$ of the national total. However, the city represents only about $4\%$ of Colombia’s population \citep{tasaHomicidios2025}. This disproportional burden reflects long-standing structural inequalities, including high unemployment rates, social and opportunity gaps, and deficiencies in educational coverage \citep{eder1}. Together, these factors sustain criminal dynamics and contribute to a widespread perception of insecurity among the population.

Although crime rates in Cali have declined since their peak in the early 2010s, levels of violence remain significantly above the national average and exhibit strong spatial and demographic concentration. In 2013, Cali accounted for more than $13\%$ of all homicides in Colombia, a figure that remains remarkably high relative to its population share \citep{tasaHomicidios2025}. Empirical evidence indicates that criminal violence in the city is closely linked to retaliatory dynamics, territorial disputes, and the persistence of organized violent groups. Reports by the Conflict Analysis Resource Center document that by 2014, at least seven organized violent groups were already operating in Cali, with a particularly strong impact on the youth population (CERAC, 2014).

%Local victimization data further illustrate the magnitude of the problem. According to the Cali Cómo Vamos citizen perception survey, by the end of 2023, $16\%$ of the respondents reported having been victims of a crime, robbery accounting for $86\%$ of reported cases \citet{pb}. Moreover, records from the Security Observatory of Santiago de Cali indicate that between 2020 and 2024, $52\%$ of homicides were attributed to revenge, with more than half of these cases occurring without any prior threats or attacks, suggesting highly volatile and unpredictable patterns of violence \citet{baseobs2020_2024}. These dynamics disproportionately affect young people, who account for more than half of homicide victims in the city \citep{observatorioCaliPatrimonio}.

In response to this scenario, local administrations have implemented a variety of public policies aimed at reducing crime and improving social conditions, particularly among youth exposed to violent environments. These initiatives combine surveillance strategies, law enforcement actions, and preventive social programs to promote social inclusion and reduce structural risk factors associated with violence (Alcaldía de Santiago de Cali, 2023b; 2024). However, despite sustained institutional efforts, criminal structures remain active and homicide rates continue to reflect deeply rooted social and territorial inequalities. 
This situation raises a fundamental policy question: how can limited public resources be allocated in a way that maximizes their impact on crime reduction while minimizing social and economic costs? Addressing this question requires analytical tools capable of capturing the dynamic and interactive nature of criminal behavior, as well as the heterogeneous effects of public interventions over time.

In this paper, we propose a mathematical framework inspired by epidemiological SIR models to study the dynamics of criminal behavior in Cali. Within this framework, crime is modeled as a population process in which individuals transition between susceptible, criminally involved, and rehabilitated states. Social interaction, exposure to violence, deterrence, and reintegration mechanisms, widely discussed in sociological and economic theories of crime \citep{sutherland1992principles,agnew1992foundation,becker1968, ibrahim2023mathematical,ibrahim2023optimal,sooknanan2017delinquent}, are explicitly represented, allowing these theoretical perspectives to be translated into a system of nonlinear differential equations. Public policies are incorporated as control variables acting on key transition rates, enabling the evaluation of preventive, punitive, and social reintegration strategies within a unified optimal control framework  \citep{Sooknanan2013AnotherWO,JuanNuno,Idisi2026}. The previous approach to this difficult social problem is currently being considered by various researchers around the world \citep{Gonzalez-Parra03042018,Mataru,Calatayud02012025,Ramponi,Idisi2026}

The main contribution of this work is twofold. First, it provides a mathematically rigorous model that captures some essential sociological mechanisms underlying the dynamics of urban crime. Second, it provides a policy-oriented analytical tool that enables decision-makers to assess trade-offs among different types of interventions under resource constraints. The resulting framework bridges sociological theory and mathematical modeling, contributing to the growing literature on quantitative approaches to social problems.

The remainder of the paper is organized as follows. Section 2 presents the social and theoretical framework of criminality in Cali, combining classical sociological theories with empirical evidence on the spatial and demographic distribution of crime in the city. Section 3 describes the main public policies and intervention programs implemented in Cali and establishes their interpretation within the modeling framework. Section 4 introduces the Crime/SIR mathe\-matical model, discusses its assumptions, and analyzes fundamental properties such as posi\-tivity, invariance, equilibrium points, and local stability. Section 5 formulates and analyzes the associated optimal control problem, characterizing cost-effective policy strategies and deriving necessary optimality conditions. Section 6 presents numerical simulations illustrating the model’s behavior and the impact of optimal policies. Finally, Section 7 concludes the paper with a discussion of the main findings, policy implications, and directions for future research.

%%%%%%%%%%%%%%%%%%%%%%%%%%%%%%%%%%%%%%%%%%%%%%%%%%%%%%%%%%%%%%%%%%%%%%%%%%%%%%
\section{A social theoretical frame of criminality in Santiago de Cali}\label{sec-2}
Crime is a complex social phenomenon that has been analyzed from multiple theoretical perspectives in sociology, criminology, and economics. From a sociological standpoint, differential association theory argues that criminal behavior is learned through sustained interaction within social environments where delinquent norms are transmitted and reinforced \citet{sutherland1992principles}. In addition, general strain theory interprets crime as a response to prolonged exposure to adverse social conditions such as poverty, unemployment, discrimination, and family conflict, which generate frustration and increase the likelihood of delinquent behavior \citet{agnew1992foundation}. From an economic perspective, rational choice theory conceptualizes crime as the outcome of a cost–benefit evaluation, emphasizing the role of incentives, deterrence, and expected punishment in individual decision-making \citet{becker1968}.

Despite their conceptual differences, these approaches converge in highlighting that criminal behavior arises from the interaction of social, economic, cultural, and institutional factors. Crime not only produces direct victims, but also erodes social cohesion, weakens trust in public institutions, and generates substantial economic and health-related costs, including loss of productivity, physical injury, psychological trauma, and premature mortality. These multidimensional impacts justify the need for analytical frameworks capable of integrating sociological mechanisms with dynamic and quantitative methods.

%%%%%%%%%%%%%%%%%%%%%%%%%%%%%%
\subsection{Crime in Cali}
Citizen security remains a central concern in Cali. Although official statistics indicate a gradual decline in homicide rates since the early 2010s, violence levels remain persistently high relative to other major Colombian cities and to the national average. This temporal evolution is illustrated in Figure \ref{fig:homicidios-linea}, which compares homicide rates in Cali, Bogotá, and Colombia over time. The figure shows that even during periods of overall national decline, Cali consistently exhibits higher rates of homicide, underscoring the structural nature of violence in the city.
\begin{figure*}[h]
\centering
    \parbox{\linewidth}{\centering\textit{
        Homicide rate: Cali, Bogotá D.C., and the national average.
    }}
    \includegraphics[width=12cm]{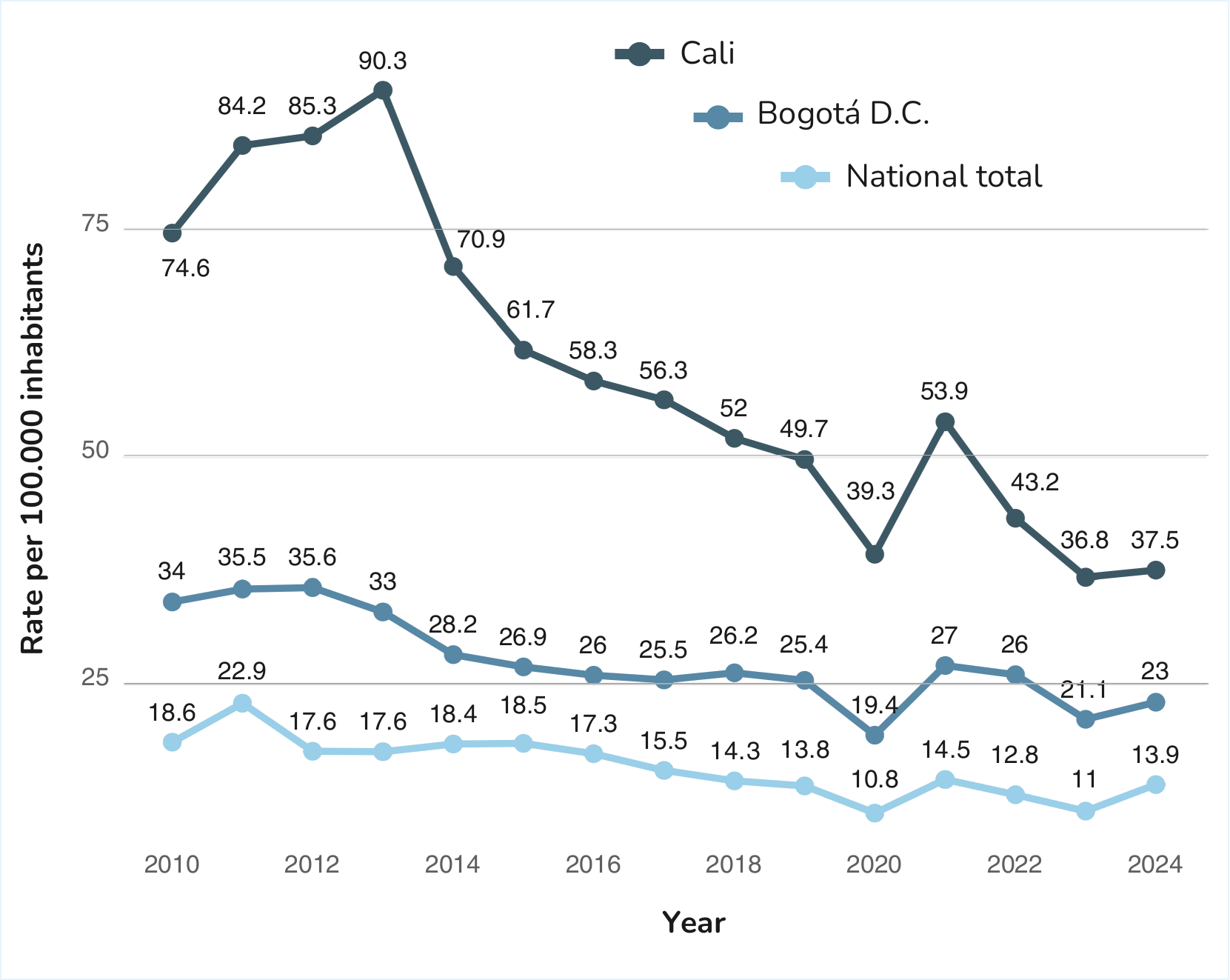}
    \caption{\footnotesize \textit{Source:} Created by the authors using data from the Ministry of Justice and Law.}   
    \label{fig:homicidios-linea}
\end{figure*}
Beyond temporal persistence, criminal violence in Cali displays a strong spatial and socioeconomic concentration. As shown in Figure 2, homicide rates are unevenly distributed throughout urban territory and are closely associated with socioeconomic strata. Neighborhoods in the lowest strata concentrate the highest levels of lethal violence, whereas areas with higher socioeconomic status exhibit substantially lower rates. This spatial pattern also reflects long-standing territorial inequalities in access to education, employment opportunities, public services, and institutional protection.

Together, Figures \ref{fig:homicidios-linea} and \ref{mapa_bivariado} highlight that violence in Cali is both temporally persistent and territorially embedded. These empirical regularities support the interpretation of crime as a dynamic social process shaped by exposure, interaction, and structural inequality rather than as isolated individual events.
\begin{figure*}[h]
\centering
 \parbox{\linewidth}{\centering\textit{
Map showing the relationship between homicide rates (H rate) and socioeconomic stratum (S stratum). In the matrix, the lowest levels of both variables appear in the lower-left corner, while the highest levels appear in the upper-right corner.}} 
    \includegraphics[width=11cm]{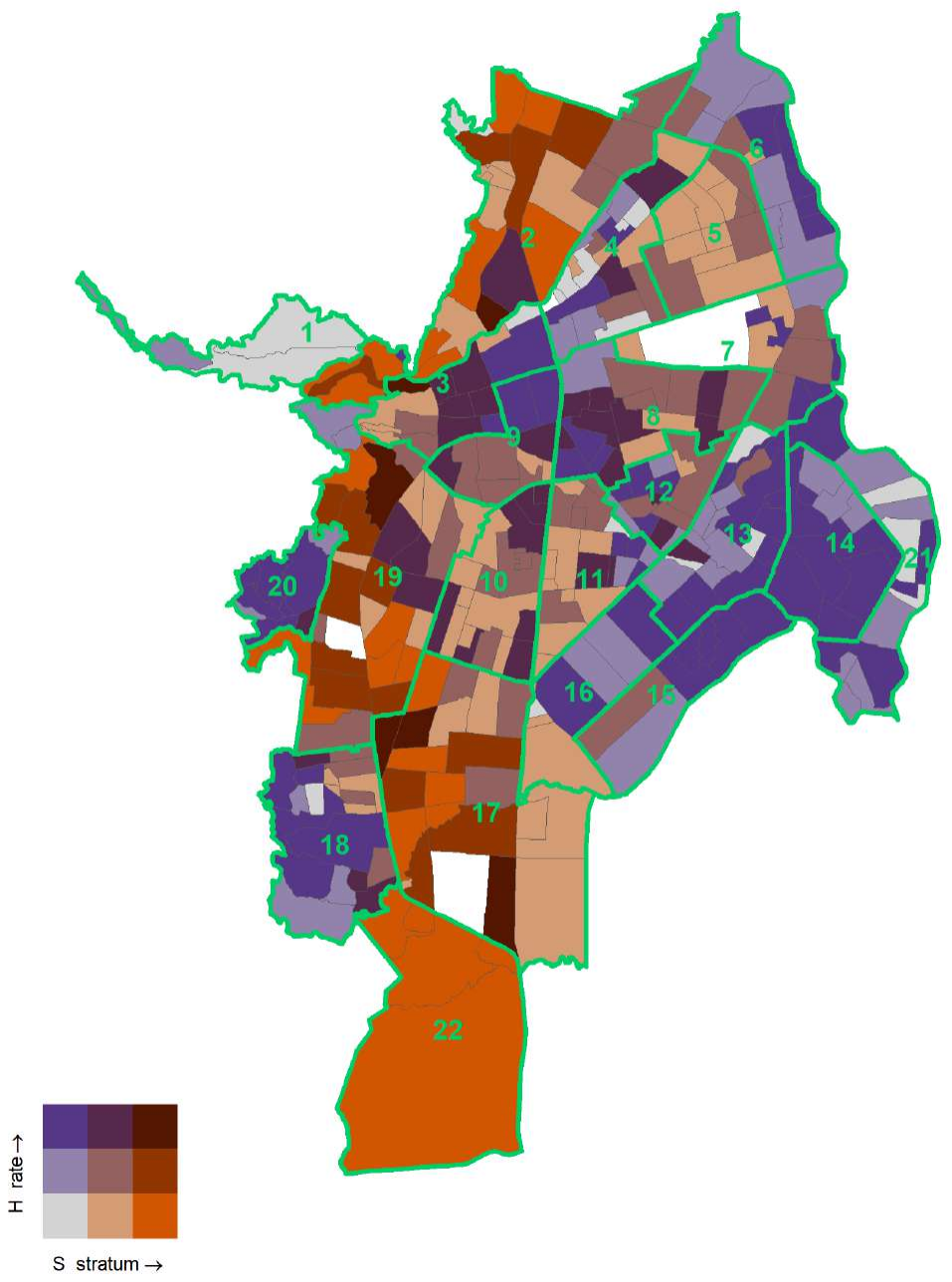}
    \caption{\footnotesize\textit{Source:} Created by the authors, using data from Santiago de Cali Security Observatory and DANE (2018)).}
    \label{mapa_bivariado}
\end{figure*}
This perspective motivates the use of population-based dynamic models. In particular, SIR-type frameworks provide a natural analogy to describe transitions between individuals who are vulnerable to criminal involvement, actively engaged in criminal behavior, and those who have left such dynamics. The corresponding compartmental structure adopted in this paper is summarized in Figure \ref{fig:diagrama compartimental} in Section \ref{sec-4}, which presents the conceptual diagram underlying the mathematical model introduced in the following section. The empirical patterns described above, namely  \emph{the temporal persistence of violence, its strong territorial concentration, and its disproportionate impact on young people}, underscore the need for public interventions that go beyond short-term enforcement strategies. In particular, they highlight the relevance of policies capable of altering the mechanisms of exposure, deterrence, and reintegration at the population level. This motivates the examination of recent public policy initiatives implemented in Cali, which serve both as institutional responses to urban violence and as empirical foundations for the mathematical model developed in this paper.

\section{Public policies and intervention programs in Cali}\label{sec-3}
In recent years, local administrations in Cali have adopted a comprehensive approach to crime reduction that combines law enforcement actions with preventive and social inclusion strategies. These policies recognize that youth participation in criminal dynamics is closely linked to structural vulnerabilities such as unemployment, school dropout, territorial exclusion, and limited access to opportunities. Consequently, interventions have increasingly focused on modifying social environments and life trajectories rather than relying exclusively on punitive measures.

Among the most relevant initiatives is the youth-oriented program En la Buena!, implemented by the Mayor’s Office of Santiago de Cali for the period 2024–2027. This program consolidates and refines strategies developed in previous administrations, integrating psychosocial support, educational reintegration, employability pathways, and community-based interventions. Its primary objective is to prevent youth violence and reduce the risk of young people being drawn into criminal dynamics, particularly in territories characterized by high levels of violence and social vulnerability.

En la Buena! targets young people who are either at risk of involvement in crime, currently involved in violent dynamics, or transitioning out of the Juvenile Criminal Responsibility System (SPRA). The program operates through a territorial approach that prioritizes communes and corregimientos with persistently high homicide rates. Interventions are implemented by interdisciplinary teams and are designed to address key risk factors such as substance use, school dropout, lack of employment opportunities, and weak social cohesion. From an analy\-tical perspective, En la Buena! provides a natural empirical framework for population-based modeling. The program explicitly distinguishes between young people who are vulnerable but not yet involved in crime, those actively engaged in violent or criminal behavior, and those undergoing rehabilitation and social reintegration processes. These categories align closely with the susceptible, involved, and recovered compartments of the Crime/SIR model introduced in the next section.

In this study, data derived from the implementation of En la Buena! are used to inform parameter selection and to validate the model numerically. This empirical grounding ensures that the proposed mathematical framework is not purely theoretical but is instead calibrated to an actual public policy context. As such, the model allows for the evaluation of alternative intervention scenarios and supports the analysis of cost-effective strategies to reduce youth involvement in crime under realistic institutional constraints. Building on this policy context, the next section introduces the Crime/SIR mathematical model. The model forma\-lizes the mechanisms of exposure, deterrence, and reintegration described above, providing a dynamic framework for analyzing the impact of youth-oriented public policies and assessing their effectiveness through numerical simulations based on real program data.

%\section{Mathematical inputs and tools}\label{sec-3}

%In the next chapter... se presentan los resultados principales del modelo formulado, los cuáles se enfocan en la construcción de un modelo epidemiológico tipo S.I.R. para modelar y analizar la dinámica de una población bajo los efectos del accionar delicuencial de un sector de la población. Posteriormente, se formula un problema de control óptimo donde se plantea un funcional objetivo que integra los costos asociados a las intervenciones de los programas del gobierno local en favor de la reducción de las acciones de violencia en la población de jóvenes. 
\section{Crime/SIR mathematical model}\label{sec-4}
In the specific case of our city, Cali, the young people included in the pathways of the En la Buena! program come from areas with high social conflict, characterized by the presence of multiple risk factors associated with the incidence of violent events.

According to the Security Observatory (2025), this initiative has included young people linked to the Adolescent Criminal Responsibility System, intending to prevent recidivism; members of ``social barrismo," \footnote{A movement that aims to transform the passion for soccer and organized fan groups into agents of positive social change by promoting coexistence, civic engagement, and community development.} seeking their integration, and vulnerable young people from communes with high rates of violence. The program aims to ``\textit{improve employability and economic stability, facilitating access to education, employment, and entrepreneurship}."

%\vspace{0.3 cm}

Based on this previous information, the following population structure for the program's youth participants is proposed:

\begin{itemize}
\item[$\dagger)$] \textbf{Susceptible youth (S)} Program participants who have not been involved in violent dynamics.
\item[$\dagger)$] \textbf{Infected youth (I)} Program participants who have been involved in violent dynamics. Some have even been arrested or faced judicial proceedings, but have not been processed through the \emph{Juvenile Criminal Responsibility System (SRPA)}
\item[$\dagger)$] \textbf{Recovered youth (R)} Program participants who have been involved in violent dynamics and have gone through the Juvenile Criminal Responsibility System.
\end{itemize}

For a fixed time period $T\in \mathbb{R}_{+}$, let us define the functions $S=S(t)$, $I=I(t)$ and $R=R(t)$, for $t\in [0,T] $ that represent the number of teenagers in the group $S$, $I$ and $R$ at time $t$, respectively. Therefore, the total number of teens at that particular time is given and denoted by $N(t)=S(t)+I(t)+R(t)$. 

%\vspace{0.3 cm}
%\noindent
In consequence, following the classic SIR epidemiology model, see Figure \ref{fig:diagrama compartimental}, we consider the following system of differential equations

\begin{figure}[h]
\centering
\begin{overpic}[width = 0.8\textwidth, tics = 5]{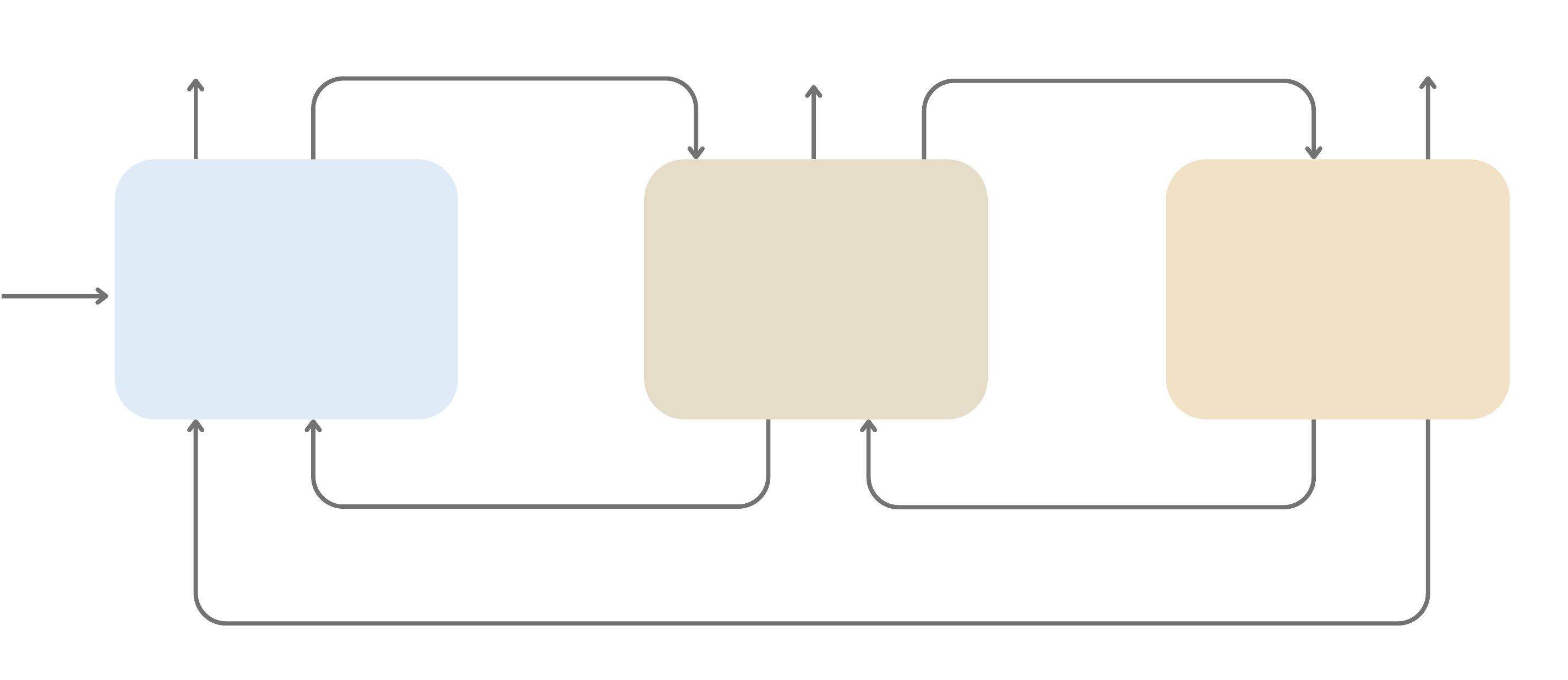}
\put(9,36.5){\large{$\phi$}}
\put(1,25.5){\large{$\Lambda$}}
\put(28,40.5){\large{$h(S)I$}}
\put(45,40.5){\large{$(\phi+\delta_1)I$}}
\put(70,40.5){\large{$\gamma_1 I$}}
\put(85,40.5){\large{$(\phi+\delta_2)R$}}
\put(49,5.1){$\rho \Omega R$}
\put(63,7.5){$(1-\Omega)\rho R$}
\put(9.5,26.5){\text{Susceptible}}
\put(13,22.5){\text{youth}}
\put(45,26.5){\text{Infected}}
\put(47,22.5){\text{youth}}
\put(77,26.5){\text{Recovered}}
\put(81,22.5){\text{youth}}
%\put(44,23.5){$Infectados$}
%\put(76.5,23.5){$Recuperados$}
\put(29.44,7.5){\large{$\gamma_2 I$}}
\end{overpic}
\caption{Compartmental diagram of the  system \eqref{SIR-criminal}}
\label{fig:diagrama compartimental}
\end{figure}

%\vspace{0.3 cm}

\begin{equation} \label{SIR-criminal}
\left\{
\begin{aligned}
    \dot{S} &= \Lambda - \frac{\alpha SI}{1+\alpha \beta S} - \phi S +\gamma_2\,I\, + \rho \Omega R, \\
    \dot{I} &=\frac{\alpha SI}{1+\alpha \beta S} - (\phi + \delta_1)I -\gamma_1\,I - \gamma_2\,I + (1-\Omega)\rho R,\\
    \dot{R} &=\gamma_1\,I - (\phi + \delta_2 + \rho)R, 
\end{aligned}
\right.
\end{equation}
%\vspace{0.5 cm}

\noindent
All parameters in \eqref{SIR-criminal} are nonnegative, and the function
\[
h(S)=\frac{\alpha S}{1+\alpha \beta S}, \quad \text{with} \quad \alpha, \beta \geq 0,
\]
It is a Holling Type II functional response in $S$, meaning that the rate at which individuals from group $I$ victimize or capture members of group $S$. The set of parameters is fully described in Table 1.

\begin{table}[h]
\centering
\caption{Description of parameters at \eqref{SIR-criminal}.}
\renewcommand{\arraystretch}{1.3} % Incrementa espacio entre filas
\begin{tabular}{@{}p{2.5 cm}p{12.5cm}@{}}
\toprule
\textbf{Parameter} & \textbf{Interpretation} \\
\midrule
$\Lambda$ & Rate of flow of adolescents at the government program.\\
$\phi$ & Attrition: rate at which youth leave the government program.\\
$\delta_1$ & Exit rate associated with risk of delinquency involvement youth desistment rate from delinquency activities within group $I$.\\
$\delta_2$ & Exit rate associated with risks in correctional facilities in the group \textit{R}. \\
$\Omega$ & The fraction of individuals from group $R$ who return to class $S$ after completing the recovery process in the \emph{SRPA}.\\
$\rho$ & Exit rate from correctional facilities. \\
$\gamma_1$ & Apprehension rate. \\
$\gamma_2$ &  Rate at which individuals from group $I$ desist from delinquent activities. \\
$\alpha$ & Delinquent group attack rate, that is, the rate at which individuals from group $I$ find a victim in group $S$ per unit of density of group $S$. \\
$\beta$ & \textit{(Handling Time)} Victimization time per unit of density of individuals from group $S$.\\
\bottomrule
\end{tabular}
\label{table:parameters}
\end{table}
While the program admits young people who may belong to any of the three groups defined in the model, conditions observed during strategy execution and analysis of the characterization database, provided by \citet{enlabuena2024}, indicate that the vast majority of participants belong to group $S$. This evidence justifies the simplification adopted in the model, where new program entries are considered to primarily originate from the susceptible population, omitting differentiated entry rates for the other two groups without compromising the structural validity of the analysis.

%\vspace{0.3 cm}

Conversely, the model includes a differential effect on the program exit rate for young people belonging to the different groups of the system ($S$, $I$ and $R$), based on a natural attrition rate denoted by $\phi$. For individuals in groups $I$ and $R$, additional risks are incorporated that make them more prone to abandoning the program or being withdrawn from it, since their participation in violent dynamics and criminal activities exposes them to greater threats to their life and physical integrity (such as brawls, personal attacks, or even retaliation).

%\vspace{0.3 cm}

Similarly, it is acknowledged that these groups ($I$ and $R$) face greater peer pressure to abandon the program, as their continued participation can represent a threat to the interests of criminal schemes. These structures can exert coercion to prevent young people from disengaging from illegal activities and integrating into the proposed intervention projects. These effects are captured in the additional rates $\delta_1$ and $\delta_2$, which are defined as positive, given that they increase the natural program attrition rate $\phi$.

\subsection{On the dynamics of the Crime/SIR model}
Let us begin the analysis of system \eqref{SIR-criminal} by defining $X:(S,I,R)^{tr}$. Under this notation, 
\[
X(t)=(S(t),I(t),R(t))^{tr}, \qquad \dot{X}(t)=(\dot{S}(t),\dot{I}(t),\dot{R}(t))^{tr},
\]
and
\begin{equation}\label{SIR-criminal-system}
\dot{X}(t)=F(X(t)), \quad \text{with} \quad F(X)=\begin{pmatrix}
\Lambda - h(S)I - \phi S +\gamma_2\,I\, + \rho \Omega R \\
h(S)I - (\phi + \delta_1)I -\gamma_1\,I - \gamma_2\,I + (1-\Omega)\rho R\\
\gamma_1\,I - (\phi + \delta_2 + \rho)R
\end{pmatrix}.
\end{equation}

%\vspace{0.3 cm}
\noindent
Define the sets
\[
\begin{split}
\mathbb{R}^3_{+}&=\left\{(x_1,x_2,x_3) \in \mathbb{R}^3:x_1,x_2,x_3\geq 0
\right\}, \\
\mathbb{R}^3_{++}&=\left\{(x_1(t),x_2(t),x_3(t)) \in \mathbb{R}^3:t\geq 0,\,\,  x_1(t),x_2(t),x_3(t)\geq 0
\right\}.
\end{split}
\]
\begin{theorem}\label{invarianza}(Positivity and Invariance)
Let $\textbf{X}_0=(S_0,I_0,R_0)\in \mathbb{R}^3_{+}$. Then there exists a unique solution $\Gamma(t):=(S(t,\textbf{X}_0),I(t,\textbf{X}_0),R(t,\textbf{X}_0))$ of \eqref{SIR-criminal-system} satisfying  
\[
S(0,\textbf{X}_0)=S_0, \quad I(0,\textbf{X}_0)=I_0, \quad  R(0,\textbf{X}_0)=R_0  \quad \text{and} \quad \Gamma(t)\in \mathbb{R}^{3}_{++},
\]
defined in some interval $t\in [0,w)$ with $0\leq w\leq \infty$. Moreover, if $\Gamma(0)=\textbf{X}_0\in \mathcal{D}$ with
\[
\mathcal{D}=\left\{(S,I,R) \in \mathbb{R}^3_{+}:S+I+R\leq \Lambda/\phi
\right\},
\]
then $\Gamma(t)\in \mathcal{D}$ for all $t\in [0,\infty)$.
\end{theorem}
\begin{proof}
Clearly, the function $F(X)$ is locally Lipschitz, then by the Picard-Lindelöf Theorem (see theo. 2.2 in \citep{teschl}), for a given initial condition $S(\tau_0)\ge 0, I(\tau_0)\ge 0$ and $R(\tau_0)\ge 0$, there exists a unique solution
\[
\Gamma(t,\tau_{0})=(S(t,S(\tau_0)),I(t,I(\tau_0)),R(t,R(\tau_0))),
\]
of \eqref{SIR-criminal-system} defined in $J=[\tau_{0},\tau_{0}+\varepsilon)$ for some $\varepsilon>0$. Now, suppose that there is $\tilde{t} \in J$ such that 
\[
S(\tilde{t},S(\tau_{0}))=0, \quad \text{and} \quad I(\tilde{t},I(\tau_{0})),\, R(\tilde{t},R(\tau_{0}))\geq 0.
\]
Then
\[
\dot{S}(\tilde{t},S(\tau_{0}))=\Lambda+ \gamma_2 I(\tilde{t},I(\tau_{0}))+\rho \Omega R(\tilde{t},R(\tau_{0}))>0,
\]
but this is a contradiction since $S(\tilde{t}
,S(\tau_{0}))=0$, in consequence $S(t,S(\tau_{0}))>0$ for all $t\in J$. Similarly, it can be shown that $I(t,I(\tau_{0}))\geq 0$ y $R(t,R(\tau_{0}))\geq 0$ for each  $t\in J$. From this point forward, and without loss of generality, let $\tau_{0}=0$ and consider the solution \eqref{SIR-criminal-system} given by $\displaystyle{\Gamma(t)=(S(t,\textbf{X}_{0}),I(t,\textbf{X}_{0}),R(t,\textbf{X}_{0}))}$ with 
\[
S(0,\textbf{X}_{0})=S_{0} \ge 0, \quad I(0,\textbf{X}_{0})=I_{0} \ge 0, \quad \text{and} \quad R(t,\textbf{X}_{0})=R_{0} \ge 0.
\]
Let $N(t)=S(t,\textbf{X}_{0})+I(t,\textbf{X}_{0})+R(t,\textbf{X}_{0})$ with $N(0)=S_{0}+I_0+R_0=N_{0}$. A direct computation shows that
\[
\begin{split}
\dot{N}(t)&=\Lambda-\phi N(t)-\delta_1 I(t)-\delta_2 R(t),\\
\dot{N}(t)&\leq \Lambda-\phi N(t), \quad \forall t\in J.
\end{split}
\]
Let $M(t)$ be the solution of the initial value problem
\[
\dot{M}(t)=\Lambda-\phi M(t), \quad M(0)=N_{0},
\]
given by 
\[
M(t)=(N_0-\Lambda/\phi)e^{-\phi t}+\Lambda/\phi.
\]
It follows directly that $\Lambda/\phi$ is a global attractor in the set of non-negative initial conditions $N_0$, furthermore
\[
(\ast) \qquad\text{If} \quad N_0 \leq \Lambda/\phi, \quad \text{then} \quad  M(t) \leq \Lambda/\phi, \quad \forall t\geq 0.
\]
and according to the theory of differential inequalities $N(t)\leq M(t)$ for all $t\in J$ (see theo. 1.3 in \citep{teschl}). The preceding results imply the following
\begin{itemize}
    \item [$\triangleright$] $N(t)$ is indefinitely extendible forwards, i.e., $N(t)$ is well-defined for all $t\in [0,\infty)$ implying directly the same conclusion on $\Gamma (t)$.
    \item  [$\triangleright$] If $N_0 \leq \Lambda/\phi$ ($(\ast)$ holds), since $N(t)\leq M(t)$ for all $t\ge 0$, the set $\displaystyle{\mathcal{D}}$
is positive invariant under the flow $F(X)$ of \eqref{SIR-criminal-system}, i.e., for all solution $\Gamma(t)=(S(t,\textbf{X}_0),I(t,\textbf{X}_0),R(t,\textbf{X}_0))$ with $\Gamma(0)=\textbf{X}_0 \in \mathcal{D}$, satisfies $\Gamma(t) \in \mathcal{D}$ for all $t\in [0,\infty)$, 
\end{itemize}
This concludes the proof.
\end{proof}
%\begin{remark}
%En adelante, vamos a suponer que la población total $N=N(t)$ satisface $\displaystyle{N(t)\leq \Lambda/\phi}$, en otras palabras, vamos a analizar la dinámica de \eqref{SIR-criminal}  sobre el conjunto invariante $\mathcal{D}$.
%\end{remark}

\subsection{Existence and local stability equilibrium states} 
In this subsection, the existence and local stability of the potential equilibrium states of system \eqref{SIR-criminal-system} are analyzed. These states, under the influence of criminal activity, can become admissible equilibria for the system. This involves characterizing the solutions of $F(X)=\textbf{0}$ in \eqref{SIR-criminal-system}, i.e., the associated nonlinear system of equations
\begin{equation}\label{equilibrium equation}
\begin{split}
    \Lambda - h(S)I - \phi S +\gamma_2\,I\, + \rho \Omega R=0, \\
   h(S)I - (\phi + \delta_1)I -\gamma_1\,I - \gamma_2\,I + (1-\Omega)\rho R=0,\\
  \gamma_1\,I - (\phi + \delta_2 + \rho)R=0.\\
\end{split}
\end{equation}
To simplify calculations, we rename the following parameters
\begin{equation}\label{new parameters}
\sigma_1:=\phi+\delta_1, \quad \sigma_2:=\frac{\gamma_{1}\rho}{\phi+\delta_2+\rho} \quad \text{and} \quad \gamma:=\gamma_{1}+\gamma_2.
\end{equation}
The system \eqref{equilibrium equation} is now equivalent to the algebraic system
\begin{equation}\label{ecuacion de equilibrios}
\begin{split}
\Lambda - \phi S +(\gamma_2+\Omega\sigma_2\ -h(S))I=0&,\\
\big(h(S)-(\sigma_1+\gamma)+(1-\Omega)\sigma_2)I=0&,\\
\end{split} \qquad \text{and} \quad R=\frac{\sigma_2 I}{\rho}.
\end{equation}
Sufficient conditions for the existence of at most two admissible equilibria of the system \eqref{SIR-criminal-system} are provided by the following result 
\begin{lemma}\label{lema equilibrios}
For all set of parameters, the system \eqref{SIR-criminal-system} admits the crime-free equilibrium 
\[
E_0= (\Lambda/\phi, 0, 0).
\]
Moreover, under the assumptions
\[
\quad \quad A_1) \qquad (\dagger) \quad \sigma_1+\gamma_1-\sigma_2>0 \quad \text{and} \quad (\underline{\dagger}) \quad  \sigma_1+\gamma-(1-\Omega)\sigma_2< h(\Lambda/\phi),
\]
or
\[
A_2) \qquad (\ddagger) \qquad h(\Lambda/\phi)<\sigma_1+\gamma-(1-\Omega)\sigma_2<\min \left\{\frac{1}{\beta}, \gamma_2+\Omega \sigma_2\right\},
\]
the system \eqref{SIR-criminal-system} admits an crime-endemic equilibrium $E_1=(\hat{S},\hat{I}, \hat{R})$, where $\hat{S}$ is the unique solution of $\displaystyle{
h(\hat{S})=\sigma_1+\gamma-(1-\Omega)\sigma_2}$ and
\[
\hat{I}= \frac{\Lambda-\phi \hat{S}}{\sigma_1+\gamma_1-\sigma_2}, \quad \text{and} \quad \hat{R}=\frac{\sigma_2\hat{I}}{\rho}=\frac{\gamma_1 \hat{I}}{\phi+\delta_2+\rho}.
\]
\end{lemma}
\begin{proof}
The admissible equilibrium points of the \eqref{SIR-criminal-system} are the solutions of \eqref{ecuacion de equilibrios}  in $\mathbb{R}^3_{+}$. It is clear that  $\textbf{E}_0=(\Lambda/\phi,0,0)$ satisfies  \eqref{ecuacion de equilibrios} for any set of parameters; that is, the crime-free equilibrium always exists, independent of all model parameters.
%\vspace{0.3 cm}
%\noindent
The existence of a crime-endemic equilibria reduces to the existence of a pair of positive numbers $\hat{S}$ y $\hat{I}$ such that
\[
\hat{I}=\frac{\Lambda-\phi \hat{S}}{h(\hat{S})-(\gamma_2+\Omega \sigma_2)}, \quad \text{with} \quad h(\hat{S})=\sigma_{1}+\gamma-(1-\Omega)\sigma_2.
\]
Under condition $(\dagger)$, it follows
\[
\sigma_1+\gamma-(1-\Omega)\sigma_2-(\gamma_2+\Omega \sigma_2)=\sigma_{1}+\gamma_1-\sigma_2>0 \quad \Rightarrow \quad \sigma_1+\gamma-(1-\Omega)\sigma_2>0.
\]
From the condition $(\underline{\dagger})$ the equation
\[
(\ast) \qquad h(S)=\sigma_1+\gamma-(1-\Omega)\sigma_{2},
\]
have only one solution $\hat{S} \in (0,\Lambda/\phi)$. Therefore, for $I=\hat{I}$, with
\[
(\ast \ast) \qquad \hat{I}=\frac{\Lambda-\phi \hat{S}}{h(\hat{S})-(\gamma_2+ \Omega \sigma_2)}=\frac{\Lambda-\phi \hat{S}}{\sigma_1+\gamma_1-\sigma_2}> 0.
\]
This proves the existence of a crime-endemic equilibria for \eqref{SIR-criminal} under assumptions $A_1)$. Now, suppose that condition $A_2)$, then
\[
0<h(\Lambda/\phi)<\sigma_1+\gamma-(1-\Omega)\sigma_2<1/\beta,
\]
thus, the equation $(\ast)$ admits a unique solution $\hat{S} \in (\Lambda/\phi,\infty)$. In addition, it is also true that
\[
\sigma_1+\gamma-(1-\Omega)\sigma_2-(\gamma_2+\Omega \sigma_2)=\sigma_{1}+\gamma_1-\sigma_2<0.
\]
and as a result, the value $\hat{I}$ given by $(\ast \ast)$ is positive, meaning
\[
\hat{I}=\frac{\Lambda-\phi \hat{S}}{\sigma_1+\gamma_1-\sigma_2}> 0,  \quad \text{with} \quad \hat{S}>\Lambda/\phi,
\]
which concludes the demonstration.
\end{proof}

\subsubsection{Local stability} Having established the existence of equilibrium points $\textbf{E}_{0}$ and $\textbf{E}_{1}$ for system \eqref{SIR-criminal-system}, our subsequent objective is to analyze their local stability properties. To achieve this, we will employ the linearized method, which involves studying the stability properties of the trivial solution $Z(t)=0$, of the associated linear system $\displaystyle{\dot{Z}(t)=A_{\ast} Z(t)}$ in each equilibria $\textbf{E}_{\ast}$ where
\[
 A_{\ast}=\begin{pmatrix}
    -h^{\prime}(S_{\ast})I_{\ast}-\phi & -h(S_{_{\ast}})+\gamma_2 & \rho \Omega\\
    h^{\prime}(S_{\ast})I_{\ast} & h(S_{\ast})-(\phi+\delta_1)-(\gamma_1+\gamma_2) & (1-\Omega)\rho\\
    0 & \gamma_1 & -(\phi+\delta_2+\rho)
\end{pmatrix}.
\]
\begin{prop}\label{crime-free equilibria stability}(Local stability of the crime-free equilibrium $\textbf{E}_{0}$)  For the system \eqref{SIR-criminal-system} and the parameters \eqref{new parameters}, the equilibrium $\textbf{E}_0$ satisfies
\begin{itemize}
    \item[$\triangleright$] $\textbf{E}_0$ is unstable if the condition $(\underline{\dagger})$ of Lema \ref{lema equilibrios} holds, i.e., 
\[
h(\Lambda/\phi)>\sigma_1+\gamma-(1-\Omega)\sigma_2,
\]
\item[$\triangleright$]$\textbf{E}_{0}$ is locally  stable if
\[
h(\Lambda/\phi)\leq \sigma_1+\gamma-(1-\Omega)\sigma_2.
\]
Moreover, if the inequality is strict, then $\textbf{E}_{0}$ is locally asymptotically stable.
\end{itemize}
\end{prop}

\begin{proof}
For $\textbf{E}_0=(\Lambda/\phi,0,0)$, the associated linear system $\displaystyle{\dot{Z}(t)=A_{0} Z(t)}$ has the matrix $A_{0}$ (rewritten with the parameters $\sigma_1,\sigma_2$ y  $\gamma$) given by 
\[
A_{0}=\begin{pmatrix}
-\phi & -h(\Lambda/\phi) + \gamma_2 & \Omega \rho \\
0 & h(\Lambda/\phi)- (\sigma_1 + \gamma)& (1-\Omega)\rho\\ 0 & \gamma_1 & - \gamma_1\sigma^{-1}_{2}\rho
\end{pmatrix}.
\]
A direct computation show that the eigenvalues $\lambda_i$, $i=0,1,2$ of $A_{0}$ are
\[
\begin{split}
\lambda_0&=-\phi, \\
\lambda_{1,2}&=\frac{\nu-\gamma_1\sigma^{-1}_2\rho\pm \sqrt{\big(\nu-\gamma_1\sigma^{-1}_2\rho\big)^2+4\big(\nu\gamma_1\sigma^{-1}_2\rho+(1-\Omega)\gamma_1\rho\big)}}{2}, 
\end{split}
\]
with $\displaystyle{\nu=h(\Lambda/\phi)-(\sigma_1+\gamma)}$. Let 
\[
\begin{split}
\Upsilon&=\nu \gamma_1\sigma^{-1}_{2}\rho+(1-\Omega)\gamma_1\rho,\\
&=\nu(\phi+\delta_2+\rho)+(1-\Omega)\gamma_1\rho.
\end{split}
\]
It is clear that
\[
I) \quad \Upsilon>0  \quad \Leftrightarrow \quad h(\Lambda/\phi)-(\sigma_1+\gamma)>-(1-\Omega)\sigma_2,
\]
which is precisely the condition $(\underline{\dagger})$ in Lemma \ref{lema equilibrios}. From here, it follows that $\lambda_{1}\lambda_2<0$, proving that $\textbf{E}_{0}$ is unstable.

%\vspace{0.3 cm}
%\noindent
Now assume that 
\[
II) \quad \Upsilon\leq 0 \quad \Leftrightarrow \quad \nu<-(1-\Omega)\sigma_{2} \quad \Leftrightarrow \quad h(\Lambda/\phi)-(\sigma_1+\gamma)\leq-(1-\Omega)\sigma_2.
\]
Then, the real part of $\lambda_{1,2}$, denoted by  $\displaystyle{\text{Re}(\lambda_{1,2})}$, satisfy  $\displaystyle{\text{Re}(\lambda_{1,2})\leq 0}$. In consequence, $\textbf{E}_{0}$ is locally stable. Further  
\[
\text{If} \quad \nu <\gamma_{1}\sigma^{-1}_{2}\rho, \quad \text{then} \quad  \text{Re}(\lambda_{1,2})<0,
\]
And therefore $\textbf{E}_{0}$ is locally asymptotically stable.
\end{proof}

\begin{remark}
If $\rho=0$ occurs simultaneously that $\textbf{E}_{1}$ exists, $\textbf{E}_{0}$ is unstable and $\hat{S}<\Lambda/\phi.$ 
\end{remark}
%\vspace{0.3 cm}
%\noindent

Concerning the local stability of $\textbf{E}_1$, the corresponding matriz $A_{1}$ is given by
\[
 A_{1}=\begin{pmatrix}
    -h^{\prime}(\hat{S})\hat{I}-\phi & -h(\hat{S})+\gamma_2 & \rho \Omega\\
    h^{\prime}(\hat{S})\hat{I}& -(1-\Omega)\sigma_{2}& (1-\Omega) \rho\\
    0 & \gamma_1 & -\gamma_1\sigma^{-1}_{2}\rho
\end{pmatrix},
\]
meanwhile, the characteristic polynomial $\displaystyle{p(\lambda):=\det(A_{1}-\lambda I_2)}$ is 
\begin{equation}\label{polinomio caracteristico de E1}
\begin{split}
p(\lambda)&=-\lambda^{3}+\tau_2\lambda^{2}-\tau_{1}\lambda+d,\\
&=-(\lambda-\tilde{\lambda}_1)(\lambda-\tilde{\lambda}_2)(\lambda-\tilde{\lambda}_3),
\end{split}
\end{equation}
where
\[
\begin{split}
\tau_{2}&=\sum_{j=1}^{3}\tilde{\lambda}_{j}=-\big(h^{\prime}(\hat{S})\hat{I}+(1-\Omega)\sigma_{2}+2\phi+\delta_2+\rho\big)<0,\\
\tau_{1}&=\sum_{i,j=1, i\neq j}^{3}\tilde{\lambda}_{j}=\phi\big(\gamma_{1}\sigma^{-1}_{2}\rho+(1-\Omega)\sigma_2,\big)+h^{\prime}(\hat{S})\hat{I}\big(\gamma_{1}\sigma^{-1}_{2}\rho+\sigma_{1}+\gamma_{1}\big)>0,\\
d &=\prod_{j=1}^{3}\tilde{\lambda}_{j}=\big(\gamma_{1}\rho-(\sigma_1+\gamma_1)\gamma_1\sigma^{-1}_{2}\rho\big) h^{\prime}(\hat{S})\hat{I},\\
&=\big(\sigma_2+\gamma_2-(\sigma_1+\gamma)\big)\hat{I} (\phi+\delta_2+\rho) h^{\prime}(\hat{S})=(\phi \hat{S}-\Lambda)(\phi+\delta_2+\rho)h^{\prime}(\hat{S}).
\end{split}
\]

\begin{prop}\label{endemic-crime equilibria stability}(Local stability of the endemic-crime equilibrium $\textbf{E}_{1}$) For the system \eqref{SIR-criminal-system} and the parameters \eqref{new parameters}, the equilibrium $\textbf{E}_1$ satisfies
\begin{itemize}
    \item [$\triangleright$] If condition  $A_1)$ in  Lema \ref{lema equilibrios} holds, then $\textbf{E}_1$ is locally asymtotically stable if and only if $\displaystyle{\tau_{2}\tau_1-d<0}.$
    \item [$\triangleright$] If condition  $A_2)$ in  Lema \ref{lema equilibrios} holds, then $\textbf{E}_1$ is unstable.
\end{itemize}

\end{prop}
\begin{proof}
Assume that condition $A_1)$ holds.
By hypothesis, we have
\[
h(\Lambda/\phi)>h(\hat{S}) \quad \Leftrightarrow \quad \Lambda/\phi >\hat{S} \quad \Leftrightarrow \quad d<0.
\]
As a consequence of \eqref{polinomio caracteristico de E1}, we deduce
\[
-p(\lambda)=\lambda^3-\tau_2 \lambda^2+\tau_1\lambda -d, \quad \text{with} \quad -\tau_2>0, \tau_{1}>0 \quad \text{and} \quad -d>0.
\]
Applying the Routh-Hurwitz criteria to $-p(\lambda)$, all their roots are negative or will have a negative real part if and only if 
\[-\tau_2 \tau_1>-d, \quad \Leftrightarrow \quad \tau_2 \tau_1-d<0,
\]
and from here, the conclusion follows directly.
%\vspace{0.3 cm} 

If condition $A_2)$ holds, we have 
\[
h(\Lambda/\phi)<h(\hat{S}) \quad \Leftrightarrow \quad \Lambda/\phi <\hat{S} \quad \Leftrightarrow \quad d>0.
\]
In consequence, from \eqref{polinomio caracteristico de E1} it follows
\[
p(0)=d>0 \quad \text{and} \quad \lim_{\lambda \to \infty}p(\lambda)=-\infty.
\]
Therefore, there exists $\tilde{\lambda}\in \mathbb{R}^{+}$ such that $p(\tilde{\lambda})=0$, which means that $A_1$ admits a strictly positive real eigenvalue showing the instability of $\textbf{E}_{1}$.
\end{proof}

\subsubsection*{\textbf{When does a persistent crime scenario emerge?}}
It is interesting to observe how the result of the instability of the crime-free equilibrium, given in Proposition \ref{endemic-crime equilibria stability}, tells us how the criminal infection (Crime) spreads through the population. Indeed, if we write system \eqref{SIR-criminal} in the form
\begin{equation} \label{SIR-criminal modificado}
\left\{
\begin{aligned}
\dot{I} &=h(S)I -\big((\phi+\delta_1)I+(\gamma_{1}+\gamma_2)I-(1-\Omega)\rho R\big),\\
\dot{R} &=-\big((\phi + \delta_2 + \rho)R-\gamma_1 I\big),\\
    \dot{S} &=  -\big(h(S)I + \phi S -\Lambda -\gamma_2\,I\, - \rho \Omega R\big),    
\end{aligned}
\right.
\end{equation}
and define the functions;
\[
\begin{split}
\mathcal{F}_{1}(I,R,S)&=h(S)I, \qquad \mathcal{V}_{1}(I,R,S)=(\phi+\delta_1)I+(\gamma_1+\gamma_2)I-(1-\Omega)\rho R,\\
\mathcal{F}_{2}(I,R,S)&=0, \hspace{1.6 cm} \mathcal{V}_{2}(I,R,S)=(\phi + \delta_2 + \rho)R-\gamma_1 I,\\
\mathcal{F}_{3}(I,R,S)&=0, \hspace{1.6 cm} \mathcal{V}_{3}(I,R,S)=f(S)I + \phi S -\Lambda -\gamma_2\,I\, - \rho \Omega R,\\
\end{split}
\]
Direct calculations show that \eqref{SIR-criminal modificado} verifies conditions a, b, and c. presented in the Appendix for system \eqref{SIR-modificado}. Furthermore, with $x_1=I$, $x_2=R$ y $x_{3}=S$ then the matrices are
\[
F=\left[\frac{\partial \mathcal{F}_{i}}{\partial x_{j}}\right]\bigg|_{(0,0,\Lambda/\phi)} \quad \text{and} \quad V=\left[\frac{\partial \mathcal{V}_{i}}{\partial x_{j}}\right]\bigg|_{(0,0,\Lambda/\phi)},
\]
are given by
\[
F=\begin{pmatrix}
h(\Lambda/\phi) & 0 & 0\\
0 & 0 & 0\\
0 & 0 & 0
\end{pmatrix} \quad \text{and} \quad V=\begin{pmatrix}
\phi+\delta_1 +(\gamma_1+\gamma_2) & -(1-\Omega)\rho & 0\\
-\gamma_1 & (\phi+\delta_2+\rho) & 0\\
h(\Lambda/\phi)-\gamma_2& -\Omega\rho & \phi
\end{pmatrix},
\]
respectively. Therefore,
\[
FV^{-1}=\begin{pmatrix}
\frac{(\phi+\delta_1+\rho)h(\Lambda/\phi)}{(\phi+\delta_1+(\gamma_1+\gamma_2))(\phi+\delta_2+\rho)-((1-\Omega)\rho\gamma_1)} & \frac{(1-\Omega)\rho \,h(\Lambda/\phi)}{(\phi+\delta_1+(\gamma_1+\gamma_2))(\phi+\delta_2+\rho)-((1-\Omega)\rho\gamma_1)} & 0 \\
0 & 0 &0 \\
0 & 0 &0 
\end{pmatrix}.
\]
The set of eigenvalues of $FV^{-1}$ is given by
\[
\left\{0,0, \frac{(\phi+\delta_1+\rho)h(\Lambda/\phi)}{(\phi+\delta_1+(\gamma_1+\gamma_2))(\phi+\delta_2+\rho)-((1-\Omega)\rho\gamma_1)}\right\}.
\]
Define the parameter
\[
\mathscr{R}_{0}=\frac{(\phi+\delta_1+\rho)h(\Lambda/\phi)}{\big|(\phi+\delta_1+(\gamma_1+\gamma_2))(\phi+\delta_2+\rho)-((1-\Omega)\rho\gamma_1)\big|}.
\]
Notice that
\[
\mathscr{R}_{0}>1 \quad \Leftrightarrow \quad (\phi+\delta_2+\rho)(h(\Lambda/\phi)-(\phi+\delta_1+\gamma))>-(1-\Omega)\rho \gamma_1, 
\]
a condition that, rewritten with the parameters \eqref{new parameters}, becomes
\[
\mathscr{R}_{0}>1 \quad \Leftrightarrow \quad h(\Lambda/\phi)>\sigma_{1}+\gamma-(1-\Omega)\sigma_{2},
\]
which is precisely the instability condition of the equilibrium $E_{0}$ (in this context written as $E_{0}=(0,0,\Lambda/\phi)$) given by Proposition \ref{crime-free equilibria stability}. Thus, the threshold value $\mathscr{R}_{0}$ helps us measure the probability that a young person with criminal behavior victimizes a susceptible young person, with $\mathscr{R}_{0}>1$ being the scenario to be avoided. 

This also indicates that the main objective of government entities and other institutions of law enforcement, security, and social welfare is to maintain control over all parameters to ensure $\mathscr{R}_{0}\leq 1$. In this case, the qualitative behavior of the system will admit only one equilibrium: specifically, the equilibrium $\textbf{E}_{0}$ and the activity of the group of young criminals will always be under control.

\section{A cost-effective and optimal population control strategy}\label{sec-5}
Based on the previous Crime/SIR model, a series of strategies are proposed to minimize crime in the city (through the control or reduction population at group $I$), using an optimal combination of controls according to the available resources and the specific socioeconomic conditions of the city. The proposed controls initially include educational programs in communities with a high incidence of crime, focusing on the prevention of criminal behavior, life skills, and employment opportunities, as for example, is currently done through the \emph{En la Buena!} government program.

In addition, surveillance strategies are also considered, which focus on areas with high rates of violence and crime, using data analysis to identify areas and times of risk. Finally, the development of co-existence spaces and community collaboration is promoted to strengthen social cohesion in neighborhoods affected by the presence of gangs and that face high levels of violence that impact community well-being, based on the recognition that recovery of the social fabric is fundamental for a sustainable and effective intervention.

%\vspace{0.5 cm} 
Based on the above, three types of controls are proposed:
\begin{itemize}
    \item[$\triangleright$] \textbf{Preventive ($u_1$):} Directed at young people who have not yet participated in criminal activities (susceptibles), but who live in environments of high vulnerability and exposure to the dynamics of violence and crime.
    \item[$\triangleright$] \textbf{Punitive/Surveillance ($u_2$):} Implemented through effective (efficient + successful) interventions of pursuit, prosecution, conviction, and incarceration of young people who have already participated in criminal activities (infected).
    \item[$\triangleright$] \textbf{Social Reintegration ($u_3$):} This type of control is directed at young people who have committed a crime and are immersed in dynamics of violence (infected), to persuade them to change their lifestyle.
\end{itemize}

Through these controls, the goal is to reduce the incidence of young people in the city joining gangs and criminal activities, either by preventing their insertion (with the policy $u_1$), effectively prosecuting wrongful acts (with the policy $u_2$), or persuading them to desist from belonging to such groups (with the policy $u_3$). Following this, the first differential equation modeling scheme is presented, which represents the interactions between individuals and how they transition from one group to another

\begin{equation} \label{Sistema de ecuaciones}
\left\{
\begin{aligned}
    \dot{S} &= \Lambda - (1-u_1) h(S)I - \phi S + (1+u_3)\gamma_2\,I\, + \rho \Omega R, \\
    \dot{I} &= (1-u_1)h(S)I - (\phi + \delta_1)I - (1+u_2)\gamma_1\,I - (1+u_3)\gamma_2\,I + (1-\Omega)\rho R,\\
    \dot{R} &= (1+u_2)\gamma_1\,I - (\phi + \delta_2 + \rho)R.
\end{aligned}
\right.
\end{equation}

%\vspace{0.3 cm} 

The following section is dedicated to identifying the combination of policies that generates the greatest possible impact in terms of reducing the number of young people involved in criminal or violent activities, while simultaneously increasing the effort to ensure adequate rehabilitation. This is done without losing sight of the goal of minimizing the cost of these interventions; that is, ensuring that their implementation is cost-effective.

%\vspace{0.3 cm} 
Based on the foregoing and denoting by $\mathscr{C}_i > 0$ the cost associated with control $u_i$, con $i = 1,2,3$  (the control policies implemented by the local government), the following optimal control problem is considered
\begin{equation}\label{problema de control óptimo}
\begin{split}
&\min_{\mathcal{U}} \mathscr{F}[u_1, u_2, u_3]= \int_{0}^{T_{\text{F}}} \left(I(t) -
R(t) + \frac{\mathscr{C}_1}{2} u_1^2(t) + \frac{\mathscr{C}_2}{2} u_2^2(t) + \frac{\mathscr{C}_3}{2} u_3^2(t) \right) \, dt,\\
&\text{s.t.}\\
    &\dot{S} = \Lambda - (1-u_1)h(S)I - \phi S + (1+u_3)\gamma_2I + \rho \Omega R, \\
    &\dot{I} = (1-u_1)h(S)I - (\phi + \delta_1)I - (1+u_2)\gamma_1 I - (1+u_3)\gamma_2I + (1-\Omega)\rho I,\\
    &\dot{R} = (1+u_2)\gamma_1 I - (\phi + \delta_2 + \rho)R,\\
\end{split}
\end{equation}
with initial conditions $S(0)=S_0$, $I(0)=I_0$ y $R(0)=R_0$, satisfying
\[
S_0,I_0, R_0 \in \mathbb{R}_{+} \quad \text{and} \quad  S_0+I_0+R_0\leq \Lambda/\phi,
\]
being $\mathcal{U}$ the set of control variables (Lebesgue-integrable functions) given by 
\[
\mathcal{U}:\left\{(u_1,u_2,u_3):u_i\in L^{1}(0,T_{F}), 0 < u_{i,m}\leq u_{i}(t)\leq u_{i,M} \right\}.
\]
with where $T_{F}<\infty$ is the final implementation time of the control policies. Following the proof of Theorem \ref{invarianza}, it is not difficult to prove that the set $\mathcal{D}\in \mathbb{R}^{3}_{+}$ given by 
\[
\mathcal{D}=\left\{(S,I,R) \in \mathbb{R}^3_{+}:S+I+R\leq \Lambda/\phi
\right\},
\]
is positively invariant, for the system in \eqref{problema de control óptimo} for all $\textbf{u} \in \mathcal{U}$. Let $\textbf{x}=\textbf{x}(t)$ and $\textbf{u}=\textbf{u}(t)$ be the state and control variables of \eqref{problema de control óptimo}, respectively, i.e., 
\[
\textbf{x}(t)=\begin{pmatrix}
    S(t)\\
    I(t)\\
    R(t)
\end{pmatrix} \quad \text{and} \quad \textbf{u}(t)=\begin{pmatrix}
    u_1(t)\\
    u_2(t)\\
    u_3(t)
\end{pmatrix},  \quad t\in [0,T_F]
\]
 From Chapter III in \citep{Fleming}, an existence result for solutions to the optimal control problem \eqref{problema de control óptimo} is derived. Firstly, for a given $\textbf{u}=\textbf{u}(t) \in \mathcal{U}$, let $\textbf{x}(t,\textbf{x}_{0})$  be a solution of the system of differential equations at \eqref{problema de control óptimo} with $\textbf{x}_{0}\in \mathcal{D}$ and consider the set
\[
\mathscr{D}=\left\{(\textbf{x}_{0},\textbf{u}): u\in L^{1}(0,T_F), \, u(t)\in \mathcal{U} \quad \text{and} \quad (0,t,\textbf{x}_{0},\textbf{x}(t,\textbf{x}_0)) \in \mathcal{D}^{\prime}_{t}=[0,T_f]\times \mathcal{D}\times \mathcal{D}\right\}.
\]
which is not empty and compact for each $t\in [0,T_F]$.

\begin{theorem} (Existence of optimal trayectories)
For the optimal control problem \eqref{problema de control óptimo}, there exists a combination of optimal trayectories 
\[
\textbf{x}_{\ast}(t,\textbf{x}_{\ast})=\begin{pmatrix}
S(t,S_{\ast})\\
I(t,I_{\ast})\\
R(t,R_{\ast})
\end{pmatrix} \quad \text{and} \quad \textbf{u}_{\ast}(t)=\begin{pmatrix}
u^{\ast}_{1}(t)\\
u^{\ast}_{2}(t)\\
u^{\ast}_{3}(t)
\end{pmatrix},
\]
that minimice the functional $\mathscr{F}(u_1, u_2, u_3)$ over the set $\displaystyle{\mathscr{D}}$. Moreover, $\textbf{x}_{\ast}(t,\textbf{x}_{\ast})$ is a $C^{1}$ function and $\textbf{u}_{\ast}(t)$  a continuos function.
\end{theorem}

\begin{proof}
The proof consists of verifying the conditions of Theorem 4.1 in \citep{Fleming} adapted to the \textit{Lagrange optimal control problem} \eqref{problema de control óptimo}. To this end, it is not difficult to check the following:
\begin{itemize}
    \item[a.] $\mathscr{D}$ is not empty, $\mathcal{U}$ and  $\mathcal{D}^{\prime}_{t}$ are compact.
    \item[b.] $F(t,\textbf{x},\cdot)$ is a linear function, i.e.,  $F(t,\textbf{x},\textbf{u})=A(t,\textbf{x})\textbf{u}+B(t,\textbf{x})$ with 
    \[
    A(t,\textbf{x})=\begin{pmatrix}
        h(S)I & 0 & \gamma_{2}I\\
        -h(S)I & -\gamma I & -\gamma I \\
        0 &\gamma_1 I & 0
    \end{pmatrix} \quad \text{and} \quad B(t,\textbf{x})=\begin{pmatrix}
        \Lambda-h(S)I-\phi S+\gamma_{2}I+\rho \Omega R\\
        (h(S)+\sigma_1+\gamma)I+(1-\Omega)\rho R\\
        -\gamma_{1} I -(\phi+\delta_2+\rho)
    \end{pmatrix}.
    \]
    \item[c.] $G(t,\textbf{x}, \cdot)$ is a convex function in  $\mathcal{U}$, indeed, the matrix
    \[
D^{2}_{\textbf{u}}G(t,\textbf{x},\textbf{u})=\begin{pmatrix}
        \mathscr{C}_1 & 0 & 0\\
        0 & \mathscr{C}_2 & 0\\
        0 & 0 & \mathscr{C}_3
    \end{pmatrix},
    \]
is positive defined. Therefore, $G(t,\textbf{x},\cdot)$ is a convex function in  $\mathcal{U}$. 
    \item[d.]  $G(t,\textbf{x},\textbf{u})\geq c\|\textbf{u}\|^2-\Lambda/\phi$. In effect, since $0\leq I,R$ and any solution of the system of equations at \eqref{problema de control óptimo} is bounded (it follows from Teorema \ref{invarianza}), then, there exists $R_{M}>0$ such that $R\leq R_{M}$. (In particular, for $\textbf{x}_{0}\in \mathcal{D}$, $R_{M}=\Lambda/\phi$). In consequence
\[
G(t,x,u)=I-R+(\mathscr{C}_1 u^{2}_1+\mathscr{C}_2 u^{2}_2+\mathscr{C}_3 u^{2}_3)/2\geq c \|\textbf{u}\|^2-R_{M}, 
\]
with $\displaystyle{2c=\max \left\{\mathscr{C}_{1},\mathscr{C}_2,\mathscr{C}_3 \right\}}.$
\end{itemize}
Conditions a., b., and c. are the hypotheses of Theorem 4.1 (in particular, Corollary 4.1) in \citep{Fleming}. Therefore, there exists $x_{\ast} \in \mathcal{D}$ tal que $(\textbf{x}_{\ast}(t,\textbf{x}_{\ast}),\textbf{u}_{\ast}(t))$ minimizes $\mathscr{F}$ for all $t\in [0,T_F]$. Furthermore, since $D^{2}_{\textbf{u}}G(t,\textbf{x},\cdot)$ is positive definite, the regularity of the optimal trajectory is derived from Corollary 6.1 in \citep{Fleming}. 
\end{proof}

\subsection{On the characterization of the optimal trajectories.} \label{subsec_optimal_trajectories}
Pontryagin's Maximum Principle (see Chapter \citep{Fleming}) is a well-known method that provides the necessary conditions on the optimal trajectories $(\textbf{x}(t,x_{\ast}),\textbf{u}_{\ast}(t))$ of the functional $\mathscr{F}$. The Hamiltonian function corresponding to \eqref{problema de control óptimo} is given by
\begin{equation}\label{hamiltoniano}
\begin{split}
&\mathcal{H}(t,\textbf{x},\textbf{u},\textbf{z})=G(t, \textbf{x}, \textbf{u},\textbf{z})+\textbf{z}\cdot F(t, \textbf{x}, \textbf{u},\textbf{z}),\\
&\mathcal{H}(t,\textbf{x},\textbf{u},\textbf{z}) = I - R+ \frac{\mathscr{C}_1}{2} u_1^2 + \frac{\mathscr{C}_2}{2} u_2^2 + \frac{\mathscr{C}_3}{2} u_3^2
+ \sum_{j=1}^{3}z_jF_{j}(t,\textbf{x},\textbf{u}),
\end{split}
\end{equation}
where $F_{j}(t,\textbf{x}(t),\textbf{u}(t))$ are given by 
\begin{equation}\label{coordenadas del campo vectorial}
\begin{split}
    F_{1}(t,\textbf{x}(t),\textbf{u}(t))&=\Lambda - (1-u_1)h(S)I - \phi S + (1+u_3)\gamma_2I + \rho \Omega R, \\
    F_{2}(t,\textbf{x}(t),\textbf{u}(t))&=(1-u_1)h(S)I - (\phi + \delta_1)I - (1+u_2)\gamma_1I - (1+u_3)\gamma_2I + (1-\Omega)\rho I,\\
    F_{3}(t,\textbf{x}(t),\textbf{u}(t)) &= (1+u_2)\gamma_1I - (\phi + \delta_2 + \rho)R.
\end{split}
\end{equation}
and $\textbf{z}=(z_{1},z_2,z_3)$ corresponds to the auxiliary variables (named co-state) which provide additional conditions for determining the optimal trajectories $\textbf{x}_{\ast}(t,\textbf{x}_{\ast})$ and $\textbf{u}_{\ast}(t).$ For \eqref{problema de control óptimo}, the \textit{transversality conditions} are $\displaystyle{\textbf{z}(t)=\textbf{0}_{\mathbb{R}^3}}$ and therefore, the associated Hamiltonian system takes the form
\begin{equation}\label{PVI hamiltonian}
\begin{split}
\dot{\textbf{x}}(t)&=\partial_{\textbf{z}}\mathcal{H}(t,\textbf{x}(t),\textbf{u}(t),\textbf{z}(t)), \hspace{4.9 cm} \textbf{x}(0)=\textbf{x}_{0},\\
\dot{\textbf{z}}(t)&=-\partial_{\textbf{x}}G(t,\textbf{x}(t),\textbf{u}(t))-\textbf{z}(t)\cdot\partial_{\textbf{x}}F(t,\textbf{x}(t),\textbf{u}(t)), \qquad \textbf{z}(T_F)=\textbf{0}_{\mathbb{R}^{3}},\\
\end{split}
\end{equation}
which, written out explicitly, is given by
\begin{equation}\label{full hamiltonian system}
  \left\{
\begin{aligned}
 \dot{S} &= \Lambda - (1-u_1) h(S)I- \phi S + (1+u_3)\gamma_2\,I\, + \rho \Omega R, \\
    \dot{I} &= (1-u_1)h(S)I - (\phi + \delta_1)I - (1+u_2)\gamma_1\,I - (1+u_3)\gamma_2\,I + (1-\Omega)\rho R,\\
    \dot{R} &= (1+u_2)\gamma_1\,I - (\phi + \delta_2 + \rho)R,\\
     \dot{z}_1 &=(1-u_1)h^{\prime}(S)I(z_1-z_2)+\phi z_1, \\
        \dot{z}_2 &=((1-u_1)h(S)-(1+u_3)\gamma_2)(z_1-z_2)+(1+u_2)\gamma_1(z_2-z_3)+z_2(\phi+\delta_1)-1,\\
        \dot{z}_3 &= 1-\rho\Omega(z_1-z_2)-z_2\rho+(\phi+\delta_2+\rho)z_3.
    \end{aligned}
    \right.
    \end{equation}
with boundary conditions
\[
\begin{split}
S(0)&=S_{0}, \quad I(0)=I_{0}, \quad \text{y} \quad R(0)=R_{0}.\\
z_{1}(T_{F})&=0, \quad z_2(T_{F})=0, \quad \text{y} \quad z_{3}(T_{F})=0.
\end{split}
\]
\textit{The optimality condition} of $\mathcal{H}$ with respect to $\textbf{u}$ invites to the computation of the nonlinear system
\begin{equation*}
D_{\textbf{u}}\mathcal{H}(t,\textbf{x},\textbf{z})=\textbf{0} \quad \Leftrightarrow \quad \frac{\partial \mathcal{H}(t,\textbf{x},\textbf{u},\textbf{z})}{\partial u_{i}}=0, \quad i=1,2,3.
\end{equation*}
given explicitly by 
\begin{equation}\label{controles}
\begin{split}
\frac{\partial \mathcal{H}(t,\textbf{x},\textbf{u},\textbf{z})}{\partial u_{1}}&= \mathscr{C}_1 u_1 + h(S)I(z_1-z_2) = 0 \quad \Rightarrow \quad u_1 = \frac{(z_2 - z_1)h(S)I}{\mathscr{C}_1},\\
\frac{\partial \mathcal{H}(t,\textbf{x},\textbf{u},\textbf{z})}{\partial u_{2}}&= \mathscr{C}_2 u_2 -\gamma_1I(z_2-z_3) = 0 \quad \Rightarrow \quad u_2= \frac{(z_2- z_3)\gamma_1 I}{\mathscr{C}_2},\\
\frac{\partial \mathcal{H}(t,\textbf{x},\textbf{u},\textbf{z})}{\partial u_{3}}&= \mathscr{C}_3 u_3 +\gamma_2 I( z_1-z_2)= 0 \quad \Rightarrow \quad u_3= \frac{(z_2 - z_1)\gamma_2 I}{\mathscr{C}_3},
\end{split}
\end{equation}
Now, the optimal control must satisfy $\displaystyle{u_{\ast}(t)\in \mathcal{U}}$ for all $t\in [0,T_F]$, therefore, each coordinate $u_{i,\ast}(t)$ of $\displaystyle{\textbf{u}_{\ast}(t)}$ is characterized by the conditions
\begin{equation}\label{controles óptimos}
\begin{split}
u_{1,\ast}(t)&=\max\left\{u_{1,m},\min\left\{u_{1,M},\frac{(z_{2,\ast}(t) - z_{1,\ast}(t))h(S_{\ast}(t))I_{\ast}(t)}{\mathscr{C}_1}\right\}\right\}, \\
u_{2,\ast}(t)&=\max\left\{u_{2,m},\min\left\{u_{2,M},\frac{(z_{2,\ast}(t) - z_{3,\ast}(t))\gamma_1I_{\ast}(t)}{\mathscr{C}_2}\right\}\right\},\\
u_{3,\ast}(t)&=\max\left\{u_{3,m},\min\left\{u_{3,M},\frac{(z_{2,\ast}(t) - z_{1,\ast}(t))\gamma_2 I_{\ast}(t)}{\mathscr{C}_3}\right\}\right\},
\end{split}
\end{equation}
where
\[
\textbf{x}_{\ast}(t,\textbf{x}_{\ast})=\begin{pmatrix}
S(t,S_{\ast})\\
I(t,I_{\ast})\\
R(t,R_{\ast})
\end{pmatrix} \quad \text{and} \quad \textbf{z}_{\ast}(t)=\begin{pmatrix}
z_{1,\ast}(t)\\
z_{2,\ast}(t)\\
z_{3.\ast}(t)
\end{pmatrix},
\]
are the solutions of the Hamiltonian system
\begin{equation}\label{PVI hamiltonian 2}
\begin{split}
\dot{\textbf{x}}(t)&=\partial_{\textbf{z}}\mathcal{H}(t,\textbf{x}(t),\textbf{u}_{\ast}(t),\textbf{z}(t)),\\
\dot{\textbf{z}}(t)&=-\partial_{\textbf{x}}G(t,\textbf{x}(t),\textbf{u}_{\ast}(t))-\textbf{z}(t)\cdot\partial_{\textbf{x}}F(t,\textbf{x}(t),\textbf{u}_{\ast}(t)),\\
\end{split} 
\end{equation}
with boundary conditions 
\[
\textbf{x}(0,x_{\ast})=\textbf{x}_{\ast}=\begin{pmatrix}
    S_{\ast}\\
    R_{\ast}\\
    I_{\ast}
\end{pmatrix} \quad \text{and} \quad z_{\ast}(T_{F})=\begin{pmatrix}
    0\\
    0\\
    0
\end{pmatrix}.
\]
All of the above constitutes Pontryagin's Maximum Principle applied to the minimization problem \eqref{problema de control óptimo}, which we describe in the following theorem.
\begin{theorem}(Necessary conditions for the optimal trajectories)
Let $\mathcal{H}:[0,T_{F}]\times \mathcal{D}\times \mathcal{U}\times \mathbb{R}^{3}\to \mathbb{R}$ given in \eqref{hamiltoniano}, the hamiltonian function associated to optimization problem given at \eqref{problema de control óptimo}. Given optimal state variables  $\textbf{x}_{\ast}(t)=(S_{\ast}(t),I_{\ast}(t),R_{\ast}(t))$ and optimal control variables  $\textbf{u}_{\ast}(t)=(u_{1,\ast}(t),u_{2,\ast}(t),u_{3,\ast}(t))$ that minimize the functional $\mathscr{F}$, then there exist auxiliary variables (co-state variables)  $\textbf{z}(t)=(z_{1,\ast}(t),z_{3,\ast}(t),z_{3,\ast}(t))$ that satisfy the system of equations
\begin{equation*}\label{half hamiltonian system}
\left\{
\begin{aligned}
     \dot{z}_1 &=(1-u_1)h^{\prime}(S)I(z_1-z_2)+\phi z_1, \\
        \dot{z}_2 &=((1-u_1)h(S)-(1+u_3)\gamma_2)(z_1-z_2)+(1+u_2)\gamma_1(z_2-z_3)+z_2(\phi+\delta_1)-1,\\
        \dot{z}_3 &= 1-\rho\Omega(z_1-z_2)-z_2\rho+(\phi+\delta_2+\rho)z_3.
    \end{aligned}
    \right.
    \end{equation*}
with boundary conditions
\[
z_{1}(T_F)=0, \quad z_{2}(T_F)=0 \quad \text{and} \quad z_{3}(T_F)=0.
\]
Moreover, the optimal control variables satisfies $\displaystyle{D_{u}\mathcal{H}(t,\textbf{x}_{\ast},\textbf{u}_{\ast},\textbf{z}_{\ast})=\textbf{0}}$, which is equivalent to 
\begin{equation*}
\begin{split}
u_{1,\ast}(t)&=\max\left\{u_{1,m},\min\left\{u_{1,M},\frac{(z_{2,\ast}(t) - z_{1,\ast}(t))h(S_{\ast}(t))I_{\ast}(t)}{\mathscr{C}_1}\right\}\right\}, \\
u_{2,\ast}(t)&=\max\left\{u_{2,m},\min\left\{u_{2,M},\frac{(z_{2,\ast}(t) - z_{3,\ast}(t))\gamma_1I_{\ast}(t)}{\mathscr{C}_2}\right\}\right\},\\
u_{3,\ast}(t)&=\max\left\{u_{3,m},\min\left\{u_{3,M},\frac{(z_{2,\ast}(t) - z_{3,\ast}(t))\gamma_2 I_{\ast}(t)}{\mathscr{C}_3}\right\}\right\}.
\end{split}
\end{equation*}
\end{theorem}
\begin{proof}
The proof follows the lines of Pontryagin's Maximum Principle, \citep{Fleming}, which have been developed in this section.
\end{proof}

%\vspace{0.3 cm}

\begin{remark} \label{relation u1 u3}
From the equations of system \eqref{controles}, the expression is derived directly
\[
u_1= \left(\frac{\mathscr{C}_3 h(S)}{\mathscr{C}_1 \gamma_2}\right)u_3 \quad \Leftrightarrow \quad \mathscr{C}_1 u_1= \left(\frac{h(S)}{ \gamma_2}\right)\mathscr{C}_3u_3. 
\]
\end{remark}

%\vspace{0.3 cm} 
As can be seen, both an increase in the victimization rate, represented by the function $h(S)$, and an increase in the cost ratio $\mathscr{C}_3/\mathscr{C}_1$ imply that control $u_1$ must increase, that is, that more must be invested in preventive programs. The reason is that, firstly, an increase in $h(S)$ indicates that the rate at which individuals in $I$ victimize those in $S$ is increasing. Secondly, a higher cost ratio suggests that preventive control ($u_1$) is relatively cheaper than reintegration control ($u_3$). Likewise,
if the rate at which individuals in group $I$ desist from criminal activity ($\gamma_2$) decreases, implying that individuals in group $I$ are less likely to abandon criminal behavior, control one should be reinforced, and therefore it makes sense for $u_1$ to increase when $\gamma_2$ decreases. In a way, analyzing the situation without controls reveals that controls would serve to reinforce initiatives that decrease the number of offenders and discourage activities that increase the number of individuals in this group.

In this sense, the relationship between $u_1$ and $u_3$ in the equation we are analyzing suggests a complementary nature between the two controls: by reinforcing social reintegration actions (control $u_3$), it becomes necessary to simultaneously strengthen preventive strategies (control $u_1$), since when an effort is made to persuade those who have already been involved in gangs to leave those environments ($u_3$), it is necessary to increase preventive efforts, that is, those that involve activities such as learning life skills for community life, motivation to attend educational centers, and employment opportunities. In this way, the aim is to prevent those who have left group $I$ from reoffending, and it is also intended that new members of group $S$ perceive greater benefits from remaining free of crime. However, this is achieved not only through repressive measures such as arrests, but also by showing new members of the group $S$ that being there brings more benefits to them in terms of their quality of life, demonstrating that an improvement in quality of life is a preferable path to recidivism or initial involvement in crime.

\section{Numerical simulations and Discussions}\label{sec-6}

This section describes the numerical strategy used to generate the results of the study. We first perform a sensitivity analysis of the uncontrolled system to identify the parameters that most strongly affect the dynamics of $S$, $I$ and $R$. We then simulate three controlled systems, each including a single control, to isolate and quantify the individual impact of preventive, punitive, and reintegration interventions. Finally, the full system with all three controls is simulated to analyze their combined effects and interactions.

The control variables \(u_1\), \(u_2\), and \(u_3\) capture the intensities of policy associated with the prevention, enforcement, and desistance mechanisms, respectively. Specifically, \(u_1\) represents the intensity of prevention campaigns implemented by government agencies and non-governmental organizations (NGOs) through television, printed media, and door-to-door outreach, among other channels; as \(u_1 \to 1\), preventive efforts reach their maximum effectiveness. The control \(u_2\) denotes the proportion of newly identified individuals in group \(I\) who are arrested and placed in juvenile detention, where \((1+u_2)\) reflects improved efforts to enforce the law and juvenile justice relative to a baseline level. Finally, \(u_3\) corresponds to the intensity of desistance and rehabilitation policies, measured by the proportion of individuals in the group \(I\) who withdraw from unlawful activities as a result of participation in government programs promoting social reintegration and law-abiding lifestyles (En la Buena! is a nice example of these kinds of programs). All controls are bounded in the interval \([0,1]\).

\subsection{Sensitivity analysis of the parameters and simulation without controls}

The sensitivity analysis is performed using elasticity measures computed via finite-difference approximations, as closed-form analytical expressions relating parameters to model outcomes are not available. Each parameter is perturbed by a small relative amount ($±1\%$), and the resulting relative changes in selected summary measures of the system trajectories are evaluated.

To avoid biases associated with adaptive time-stepping, all numerical solutions are interpolated onto a common time grid, and summary measures are computed as the area under the curve (AUC) using the trapezoidal rule. Parameter elasticities are approximated using centered finite differences, which reduce truncation error and provide more stable estimates in the presence of nonlinear dynamics.

%The sensitivity analysis of the uncontrolled dynamic system \eqref{SIR-criminal}, reveals that the parameters $\alpha$ and $\beta$, which are associated with the Holling Type II function that models the dynamics between $S$ and $I$, exert a particularly strong influence on the evolution of $S$, $I$, and $R$. In addition to these, other parameters such as $\phi$ representing the rate at which susceptible individuals $S$ exit the system (interpreted in some models as a mortality rate), as well as $\gamma_1$, the detention rate and $\gamma_2$, the rate at which individuals in compartment $I$ are free from gang or delinquent behavior and reintegrate into the susceptible population, also play a significant role in shaping the dynamics of the system.

\begin{figure}[ht]
\centering
\includegraphics[width=0.62\textwidth]{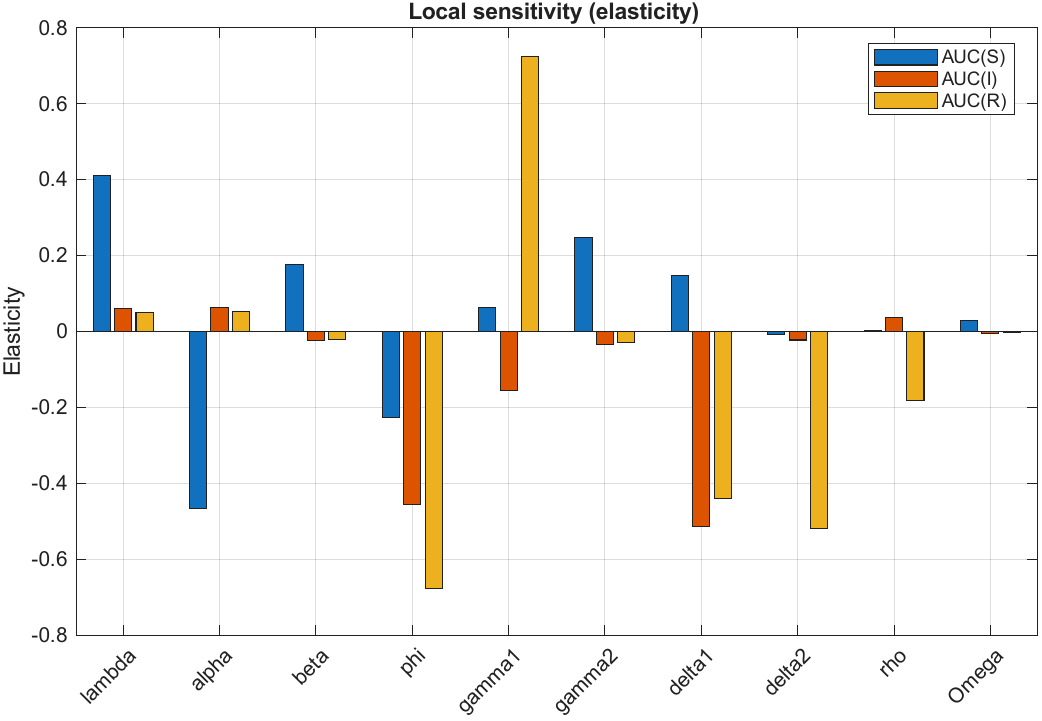}
\caption{Parameter elasticities of the Crime/SIR model.
Elasticities are computed using centered finite differences based on small relative parameter perturbations ($\pm 1\%$). Model responses are summarized by the area under the curve (AUC) of the corresponding state variables, after interpolating all numerical solutions onto a common time grid. Elasticity values are dimensionless and quantify the relative sensitivity of model outcomes to parameter variations.}
\label{fig parameter sensitivity}
\end{figure}

The values used in this numerical exercise are based on data from the En la Buena! program, provided by the Cali Security Observatory (see Table \ref{table:parameters} below). The data set includes a total of $1,539$ young individuals, of whom $1,357$ are classified as susceptible $S$, $136$ as infected $I$, and the remainder as recovered $R$.
% desde acá está la parte que yo agregué

Figure \ref{fig parameter sensitivity} reports the local sensitivity elasticities of the areas under the curve (AUC) for the Susceptible ($S$), Criminal ($I$), and Recovered ($R$) populations with respect to marginal changes in model parameters. Because elasticity is dimensionless, they allow for a direct comparison of the relative importance of parameters and their links to the policy controls $u_1, u_2, u_3$ considered. We now examine the parameters directly related to three controls variables.

Parameters $\alpha$ and $\beta$, play a central role in the system’s cumulative dynamics. Increases in effective criminal contact through $\alpha$, reduce AUC($S$) and increase AUC($I$), while an increase in the saturation parameter $\beta$ slightly increases AUC($S$) and reduces AUC($I$). Overall, these results highlight the relevance of preventive control in limiting criminal activity and preserving the susceptible population. 

The parameter $\gamma_1$, shows a negative elasticity for AUC($I$) and a positive one for AUC($R$), indicating that punitive measures reduce the cumulative time spent in criminal activities, though with limited effects on AUC($S$). Finally, $\gamma_2$, exhibits positive elasticities for AUC($S$) and negative ones for AUC($I$), suggesting that rehabilitation and social reintegration reduce cumulative criminal involvement while expanding the susceptible population.

Taken together, the sensitivity analysis indicates that sustained crime reduction requires a combination of preventive, punitive, and rehabilitative policies, underscoring their complementarity in effective public policy design.

%hasta acá 

Given the importance of the Holling type II function in our system, we want to show the behavior of our system imposing different values of the parameter $\beta$, which governs transmission from the susceptible compartment to the infected compartment. The analysis indicates that when beta is low, the dynamics is dominated by the infected group $I$. In this context, susceptible individuals (representing vulnerable youth at risk of becoming involved in criminal activity) are rapidly converted into infected individuals, modeled as those actively engaged in delinquent or gang-related behavior.

The model shows that for very low values of $\beta$; this conversion occurs at an accelerated rate. Specifically, when $\beta =2$, the susceptible population persists until approximately period 15; when $\beta=1$, it disappears by period 3; when $\beta= 0.5$, by period $1.5$; and for $\beta \leq 0.3$, the susceptible group is effectively depleted immediately. This rapid exhaustion of the susceptible pool illustrates the severity of contagion-like dynamics in environments with high rates of criminal influence. Additionally, the parameter $\Omega$, which represents the fraction of individuals returning to the susceptible group after release from the juvenile penal system (SRPA), is assumed to be small. Under these conditions, the infected population becomes dominant, implying that the community is effectively under the control of criminals ($I$ individuals).
Such dynamics are characteristic of high-crime communities, where intervention is essential. Therefore, it is crucial to implement control strategies that aim to reduce the prevalence of infected individuals and increase the proportion of recovered individuals $R$ who reintegrate into the susceptible group $S$ through rehabilitation and desistance processes. Given the constraints on public funding, these interventions must be designed to be effective and economically viable. 

\begin{figure}[ht]
\centering
\includegraphics[width=\textwidth]{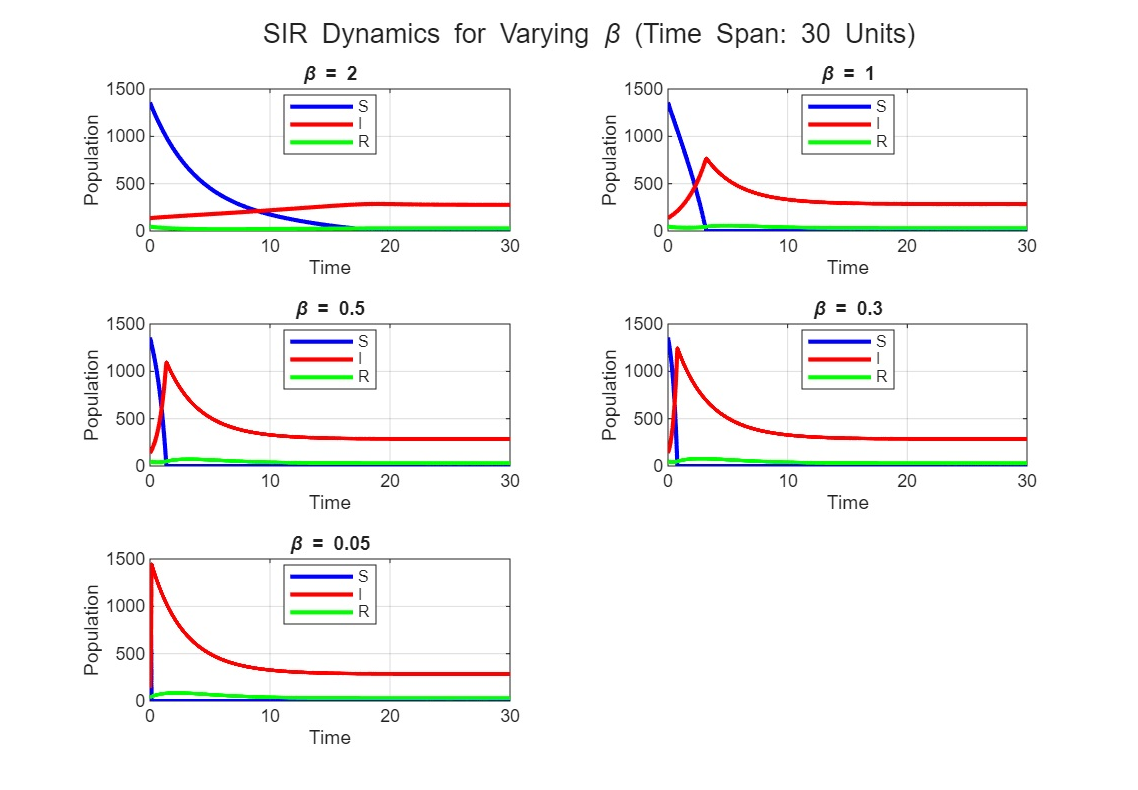}
\caption{Numerical simulation of the uncontrolled system \eqref{SIR-criminal} for different values of $\beta$.}
\label{fig roots linear}
\end{figure}

\begin{table}[ht]
\centering
\caption{Description of parameters in system \eqref{SIR-criminal}.}
\renewcommand{\arraystretch}{1}
\begin{tabular}{@{}p{3cm} p{4cm} p{2.5cm}@{}}
\toprule
\textbf{Parameter} & \textbf{Value} & \textbf{Reference} \\
\midrule
$\Lambda$ & \(100\) & Assumed \\ 
$\phi$ & \(0.27\) & COS \\ 
$\delta_1$ & \(0.05\) & Assumed \\ 
$\delta_2$ & \(0.02\) & Assumed \\ 
$\Omega$ & \(0.3\) & ICBF \\ 
$\rho$ & \(0.2\) & Assumed \\ 
$\gamma_1$ & \(0.05\) & COS \\ 
$\gamma_2$ & \(0.1\) & IGP \\ 
$\alpha$ & \(0.4\) & IGP \\ 
$\beta$ & \{2, 1, 0.5, 0.3, 0.05\} & Assumed \\ 
\bottomrule
\end{tabular}
\label{table:parameters}

\vspace{0.5cm}
\begin{flushleft}
\small
\textbf{Notes:} COS stands for \textit{Cali's Security Observatory}. IGP refers to the article entitled \textit{“Influence of Peer Groups on Adolescents at the Social Support Foundation Involved in Psychoactive Substance Use and Offenses”}. ICBF is \textit{Colombian Institute for Family Welfare (Instituto Colombiano de Bienestar Familiar)}.
\end{flushleft}
\end{table}

\subsection{Simulations with one control at a time} The purpose of this section is to analyze the impact of each control separately. Following the results obtained in the previous section, first of all we conducted a simulation applying only control $u_1$ with an associated cost of $5$, which is the control that impacts the $\alpha$ and $\beta$ parameters as specified in system \ref{full hamiltonian system}, while maintaining all parameters consistent with the previous exercises, except for beta, which was set to 0.3 (a value corresponding to a significant capture rate, see Figure \ref{fig roots linear}). The Figure \ref{with control} represents the trajectories of the state variables $S$, $I$, and $R$ over time, comparing the cases with and without the control $u_1$. This control represents the intensity of prevention or social intervention programs, which reduce the rate of transition from susceptible individuals to criminals. For this simulation, the parameter $\alpha$, which it is related to the probability that a susceptible individual will come into contact with a criminal and engage in criminal behavior, is modeled as a decreasing function of cumulative control, $\alpha_{eff}(t) = \alpha_0 e^{-ku_{cum}(t)}$, where $u_{cum}(t)$ represents the weighted sum of intervention efforts over time. This reflects the hypothesis that prevention and rehabilitation policies can produce persistent structural effects on social dynamics, progressively reducing the propensity for crime even after control efforts have diminished.
\begin{figure}[ht]
\centering
\includegraphics[width=0.7\textwidth]{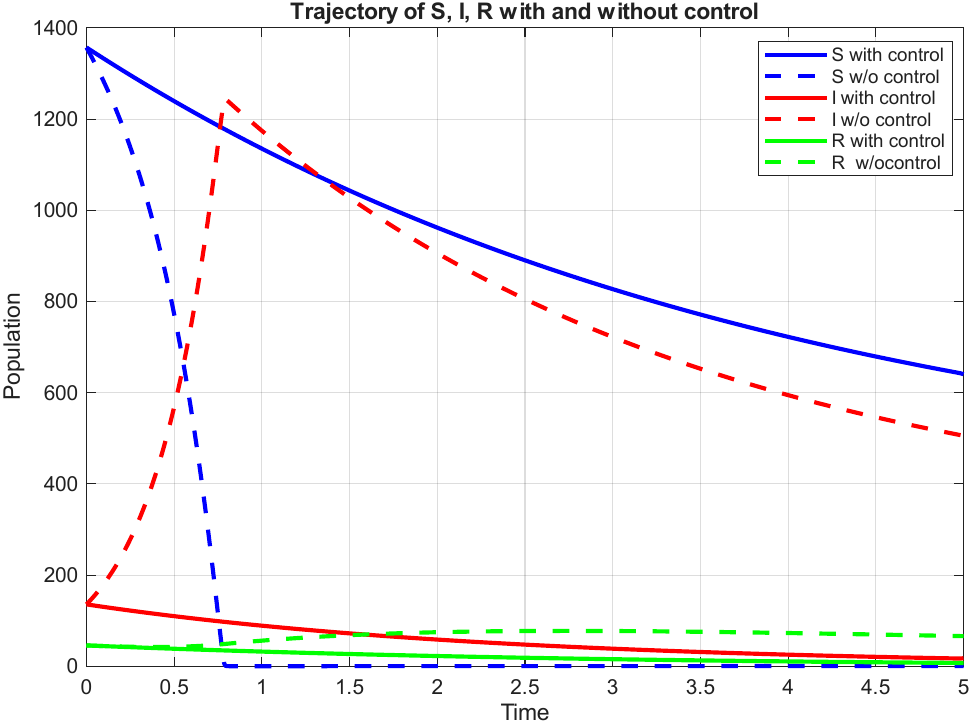}
\caption{Numerical Simulation for $\beta=0.3$ with and without control $u_1$}
\label{with control}
\end{figure}
It should be noted that, despite the relatively low beta value and the more permissive modeling of the parameter $\alpha$, the application of the control $u_1$ effectively reduces the number of infected individuals $I$ and promotes an increase in the susceptible population $S$, consistent with the expected outcomes of a policy aimed at penalizing criminal behavior. Additionally, the analysis includes a five-period horizon, reflecting the typical duration of a local administration (four years), with the extra period included to observe the behavior of the system beyond the immediate policy timeframe. It should be noted that the cost of implementing the control is relatively small compared to the objective value without control. The cost difference is approximately a factor of 10.
%This suggests that the control strategy is not only necessary but also economically viable. More precisely, the value of the objective functional is $202.3121$ with control, compared to $3366.1138$ without control.

Next we conducted a simulation applying only control $u_2$ with an associated cost of $0.1$, which impacts the proportion of newly arrested youths, reflecting the additional effort the local government is making to prevent crime through arrests. In Figure \ref{with withou control u2}, we see that increasing $(1 + u_2)\gamma_1$ accelerates the movement of individuals from $I$ to $R$. However, $S$ is barely affected directly, but its relative level can increase as $I$ decreases. We see that this control is weak in its intention to increase the number of $S$. However, a rehabilitation effort that decreases the recidivism rate denoted by $(1-\Omega)$ would generate a positive impact on the number of $S$. This implies that a strong capture policy is insufficient if the rehabilitation process leading to non-recidivism is inadequate. At this point, we want to emphasize that despite the low direct effectiveness of this control in terms of increasing $S$, it is important to maintain it in conjunction with other measures for it to ultimately be effective.

\begin{figure}[ht]
\centering
\includegraphics[width=0.74\textwidth]{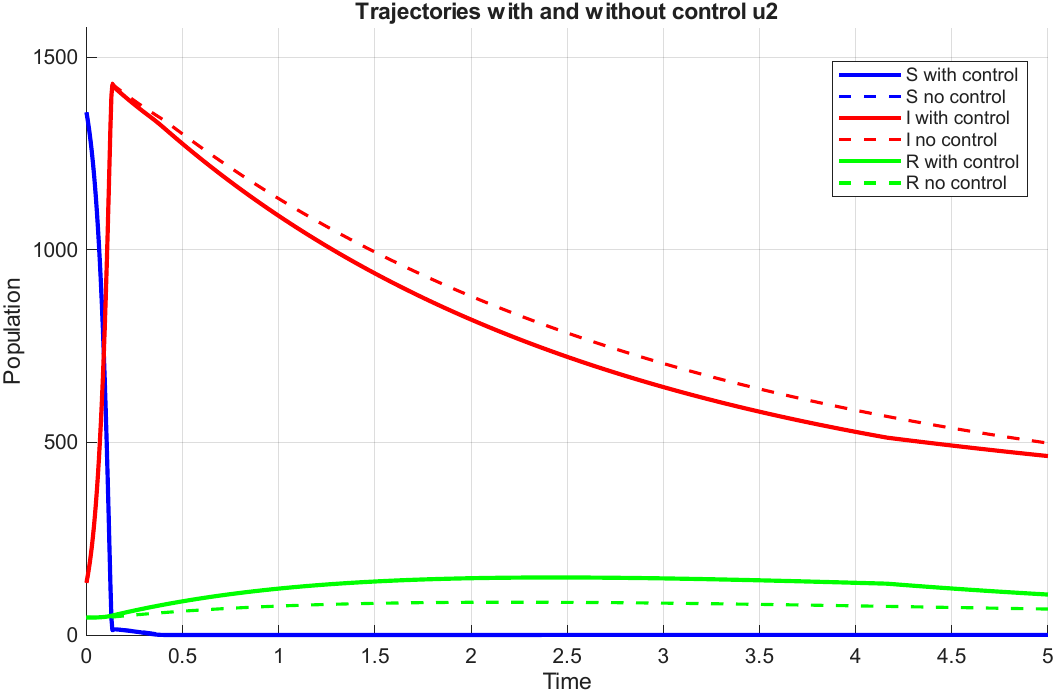}
\caption{Numerical Simulation with and without control $u_2$.}
\label{with withou control u2}
\end{figure}

Finally, we conducted a simulation applying only control $u_3$ with an associated cost of $0.1$, which impacts social reintegration due to the effect of rehabilitation programs. In Figure \ref{with withou control u3} we see that having a low $\beta$ value means that even though we have control $u_3$ moving young people from $I$ to $S$, what we are actually doing by applying control $u_3$ without considering the other controls is fueling the predatory power of population $I$ through $\alpha$ and $\beta$. Therefore, what happens is that while the control $u_1$ is inactive, the power of the control $u_3$ to increase the $S$ population is neutralized by the $I$ population. This effect was already observed theoretically (see remark \ref{relation u1 u3} at the end of subsection \ref{subsec_optimal_trajectories}), where we mentioned the complementary nature of $u_1$ and $u_3$, which seems logical at first glance, but it is interesting that we can observe the theoretical result through numerical simulations.

\begin{figure}[ht]
\centering
\includegraphics[width=0.74\textwidth]{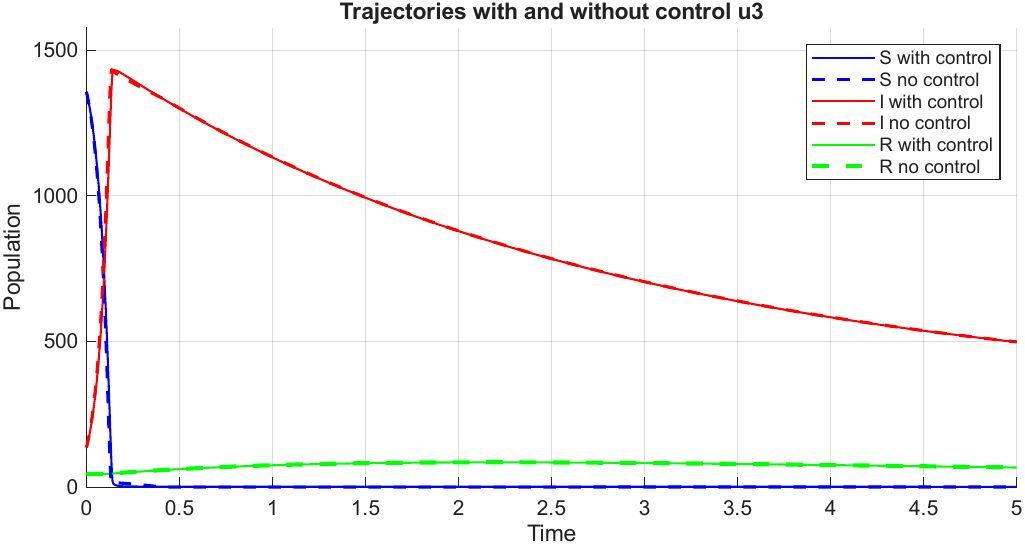}
\caption{Numerical Simulation with and without control $u_3$.}
\label{with withou control u3}
\end{figure}

\subsection{Simulation of the fully controlled system}
We conducted a simulation applying all three controls simultaneously, which significantly reduces the trajectory of $I$ compared to the uncontrolled scenario, increases the susceptible population $S$, and reduces the accumulation in the correctional population $R$, as can be seen in Figure \ref{with without all controls}.

\begin{figure}[ht]
\centering
\includegraphics[width=0.7\textwidth]{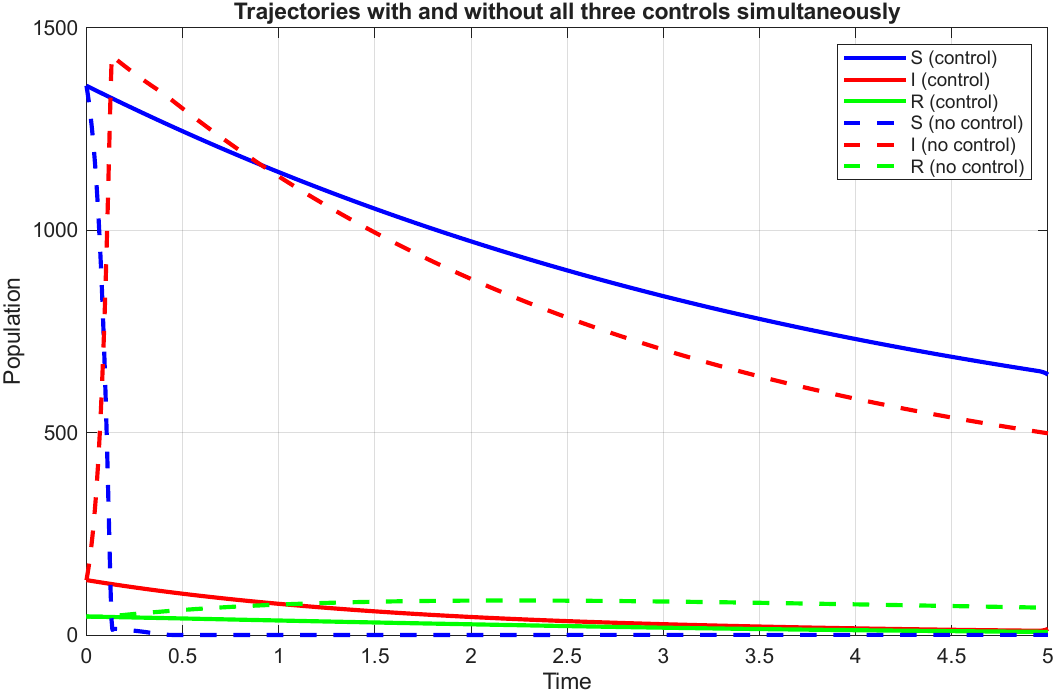}
\caption{Numerical Simulation with and without all three controls $u_1$, $u_2$ and $u_3$.}
\label{with without all controls}
\end{figure}

The results show that an integrated strategy based on prevention, surveillance, and social reintegration significantly reduces the number of young people involved in criminal activities. Prevention $u_1$ addresses the root of the problem: it decreases the likelihood that at-risk youth will join gangs or groups that commit crimes, thus reducing the social transmission of criminal behavior. This translates into fewer young people exposed to criminal dynamics and reduces future pressure on the juvenile justice system. Surveillance and punitive control measures $u_2$ complement this effect by reducing the time a young person remains involved in criminal activities and limiting their ability to negatively influence other young people. This control generates visible impacts in the short term, decreasing active criminal behavior and increasing institutional capacity to respond. At the same time, it prevents criminal networks from expanding through imitation or recruitment. Finally, social reintegration programs $u_3$ have a key medium- and long-term effect by reducing recidivism and allowing young people who have engaged in criminal activity to successfully return to a social and productive life. This not only decreases the number of offenders but also increases the number of young people actively participating in initiatives that benefit the community.

Taken together, the three types of intervention generate a sustained reduction in juvenile crime, with cumulative positive effects that could not be achieved if isolated policies were implemented. This is demonstrated by the separate implementation of the controls, which show that the three controls reinforce each other. Arrests and convictions without creating opportunities for reintegration into social and work life do not have a significant impact on young people or their environment. Similarly, social reintegration programs should be complemented by prevention programs that reduce effective contact between individuals in groups $I$ and  $S$.

\section*{Conclusions}
This study develops a population-based Crime/SIR model coupled with an optimal control framework to study youth-related criminal dynamics and public policy interventions in urban contexts. The model formalizes key sociological mechanisms of social exposure, deterrence, and reintegration within a system of nonlinear differential equations, and provides analytical results on positivity, invariance, equilibrium existence, and local stability. 

The optimal control formulation allows for a systematic assessment of preventive, punitive, and social reintegration policies under resource constraints, highlighting the role of integrated strategies in reducing criminal involvement. Numerical simulations support the analytical findings and illustrate how different policy configurations affect the system’s trajectories. Although the framework is general and applicable to different cities, its practical relevance is demonstrated through numerical validation using data from En la Buena!, a youth-oriented government program currently implemented by the Municipality of Santiago de Cali. This application shows how mathematical modeling and optimal control can support evidence-based policy design and evaluation. Future work may extend the model to include spatial heterogeneity, individual-level variability, or network effects to further enhance its applicability.
%\textcolor{red}{considero que acá sería muy importante mencionar que este trabajo podría ser un insumo para lo del crimeanlyzer y para la aplicación de ecuaciones diferenciales con retardo
One of the key findings of this study is the complementarity between controls $u_1$ and $u_3$. Specifically, when efforts are made to persuade individuals already involved in gangs to leave these environments $u_3$, it becomes necessary to reinforce preventive actions aimed at reducing the probability that individuals who have not yet engaged in criminal activities enter such pathways, as well as reducing recidivism among youth who have undergone rehabilitation in correctional facilities.

Identifying these complementarities strengthens the theoretical structure of the model by showing that crime dynamics is not driven by independent policy channels but rather by interacting mechanisms with feedback effects. This result highlights the importance of coordinated intervention strategies, suggesting that effective crime reduction requires the simultaneous implementation of preventive and rehabilitative policies rather than isolated efforts.

\section*{Appendix}
The main purpose of this appendix is to briefly present the results of Chapter 5, Section 3 of the book \citep{Maia}, inspired by the technique developed in \citep{VANDENDRIESSCHE200229} to calculate the basic reproduction number $\mathscr{R}_{0}$ of a general compartmental system. Let us consider a particular case of an SIR model, expressed (a bit differently from the traditional form) by the following system of equations
\begin{equation}\label{SIR-modificado}
\left\{
\begin{aligned}
     \dot{I}(t) &=\mathcal{F}_{1}(I(t),R(t),S(t))-(\mathcal{V}^{-}_{1}(I(t),R(t),S(t))-\mathcal{V}^{+}_{1}(I(t),R(t),S(t))),\\
     \dot{R}(t) &= \mathcal{F}_{2}(I(t),R(t),S(t))-(\mathcal{V}^{-}_{1}(I(t),R(t),S(t))-\mathcal{V}^{+}_{2}(I(t),R(t),S(t))),\\
      \dot{S}(t) &= -(\mathcal{V}^{-}_{3}(I(t),R(t),S(t))-\mathcal{V}^{+}_{3}(I(t),R(t),S(t))),
\end{aligned}
\right. 
\end{equation}

%\vspace{0.3 cm}

where $\mathcal{F}_{1}$ y $\mathcal{F}_{2}$, represent the rates of new infections in groups  $I$, $R$  and $S$ respectively, while $\mathcal{V}_1$, $\mathcal{V}_2$ and $\mathcal{V}_3$ incorporate the remaining transient terms, namely, births, deaths, disease progression, and recovery in the corresponding groups $I$, $R$ y $S$. Furthermore, the following conditions hold:
\begin{itemize}
    \item[a.] $\mathcal{F}_{i}(I,R,S)\geq 0$, $\mathcal{V}^{-}_{i}(I,R,S)\geq 0$ \quad \text{and} \quad $\mathcal{V}^{+}_{i}(I,R,S)\geq 0$ for all, $I,R,S\geq 0.$ 
     \item[b.] $\mathcal{V}^{-}_{1}(0,R,S)=\mathcal{V}^{-}_{2}(I,0,S)=\mathcal{V}^{-}_{3}(I,R,0)=0$. In particular
     \[
     \mathcal{V}^{-}_{1}(0,0,S)=\mathcal{V}^{-}_{1}(0,0,S)=0, \quad \text{for} \quad i=1,2. 
     \]
    \item[c.] $\mathcal{F}_{i}(0,0,S)=0$ y $\mathcal{V}^{+}_{i}(0,0,S)=0$ \text{for}  $i=1,2$.
\end{itemize}
Now, assume that the system 
\[
\dot{S}=-\big(\mathcal{V}^{-}_{3}(0,0,S)-\mathcal{V}^{+}_{3}(0,0,S)\big),
\]
admits a unique equilibrium state $\mathscr{E}_{\dagger}=(0,0,S_{\dagger})$ (the infection-free equilibrium), such that all solutions with initial condition $(0,0,S_{0})$ approach $\mathscr{E}_{\dagger}$ as $t\to \infty.$ Compute the matrices
\[
F=\left[\frac{\partial \mathcal{F}_{i}}{\partial x_{j}}\right]\bigg|_{(0,0,S_{\dagger})} \quad \text{and} \quad V=\left[\frac{\partial \mathcal{V}_{i}}{\partial x_{j}}\right]\bigg|_{(0,0,S_{\dagger})},
\]
with $\mathcal{F}_{3}(I,R,S)=0$, $\mathcal{V}_{i}=\mathcal{V}^{-}_{i}-\mathcal{V}^{+}_{i}$ and  $x_{1}=I$, $x_{2}=R$. The matrix $K=FV^{-1}$, is defined as the next-generation matrix (see Chapter 5, \citep{Maia}). The basic reproduction number $\mathscr{R}_{0}$ for system \eqref{SIR-modificado} is given by
\[
\mathscr{R}_{0}=\rho(FV^{-1}),
\]
where $\rho(C)$ is the spectral radius of matrix  $C$, that is
\[
\rho(C):=\sup \left\{|\lambda|:\lambda\in \sigma(C)\right\},
\]
with $\displaystyle{\sigma(C)}$  being the set of eigenvalues of $C$.

%\section*{Acknowledgments}

%Authors want to thank the financial support provided by ICETEX, which made this research possible under the program \emph{Convocatoria Subvenciones para Proyectos de Internacionalización 2024}

\bibliographystyle{plainnat}
\bibliography{bibliocrimesir}
\end{document}